\documentclass[11pt, DIV=9]{scrartcl}
\usepackage[a4paper,margin=1in,tmargin=1in,bmargin=1.25in]{geometry}

\usepackage[T1]{fontenc}
\usepackage[english]{babel}
\usepackage{amsmath, amsfonts, amsthm, amssymb, mathtools, stmaryrd}
\usepackage[pdfusetitle,hidelinks]{hyperref}
\usepackage[capitalise,noabbrev]{cleveref}
\usepackage{csquotes}
\usepackage[osf,sc]{mathpazo}
\usepackage{relsize}
\usepackage{microtype}
\usepackage{pdfpages}

\usepackage{pgfplots}

\usepackage{thm-restate}
\usepackage{subcaption}

\usepackage{xcolor}

\usepackage{graphicx}

\usepackage{complexity}

\usepackage{abstract}
\usepackage[inline]{enumitem}
\usepackage{booktabs}
\usepackage{todonotes}

\usepackage[linesnumbered,ruled,commentsnumbered,longend,]{algorithm2e}
\usepackage{pifont}
\usepackage{tcolorbox}
\SetKwProg{Fn}{function}{}{end}
\SetKwInOut{Input}{\hspace*{1.2em}Input\hspace*{0em}}
\SetKwInOut{Output}{\hspace*{0.4em}Output\hspace*{0em}}
\SetKwInOut{Auxiliary}{\hspace*{-0.5em}Auxiliary\hspace*{0em}}
\newcommand{\inputb}[1]{\indent\raisebox{-1pt}{\scalebox{1}[1]{\ding{226}}} #1\\}
\newcommand{\outputb}[1]{\indent\raisebox{-1pt}{\scalebox{-1}[1]{\ding{226}}} #1\\}

\makeatletter
\definecolor{codebg}{cmyk}{0,0,0,0}
\def\@algocf@pre@ruled{\begin{tcolorbox}[colback=codebg,arc=0.3em,boxsep=0em,left=0em, right=0.5em, top=0.5em, bottom=0.5em, boxrule=0.1mm]}%
\def\@algocf@post@ruled{\end{tcolorbox}\vspace*{-1.5em}}%
\renewcommand{\fnum@algocf}{\hspace*{0.5em}\AlCapSty{\AlCapFnt\algorithmcfname} \arabic{algocf}}
\makeatother

\usepackage{todonotes} %disable

\usepackage{tikz}
\usepackage{tikz}
\usepackage{tikz-layers}
\usetikzlibrary{calc}
\usetikzlibrary{decorations.pathmorphing}
\usetikzlibrary{positioning, arrows.meta, fit}
\usetikzlibrary{shapes}
\usetikzlibrary{arrows.meta,backgrounds,automata}
\usetikzlibrary{matrix,calc}
\tikzstyle{vertex}=[circle, draw, fill=gray!80!white,thick,scale=1.2]
\tikzstyle{edge}=[draw=black, thick,-]
\usetikzlibrary{hobby,arrows,decorations.pathmorphing,backgrounds,positioning,fit,petri,calc}
\tikzstyle{vertex}=[anchor=center, circle, fill=gray, inner sep=1.75]

\usetikzlibrary{cd}
\usepackage{microtype}

\theoremstyle{definition}
\newtheorem{definition}{Definition}[section]

\newtheorem{remark}[definition]{Remark}

\theoremstyle{plain}
\newtheorem{lemma}[definition]{Lemma}

\Crefname{fact}{Fact}{Facts}

\DeclareMathOperator{\supp}{supp}
\DeclareMathOperator{\IR}{IR}
\DeclareMathOperator{\LL}{LL}

\tikzstyle{place}=[circle,draw=black!80,thick,fill=black!80, inner sep=0pt,minimum size=1.5mm]

\renewcommand\phi\varphi
\renewcommand\epsilon\varepsilon

\DeclareMathOperator{\Sym}{Sym}

\DeclareMathOperator{\fix}{expand}
\DeclareMathOperator{\transpose}{transpose}

\DeclareMathOperator{\Aut}{Aut}
\DeclareMathOperator{\Var}{Var}

\DeclareMathOperator{\Lit}{Lit}

\DeclareMathOperator{\dcup}{\dot{\cup}}

\newcommand{\Orbi}[2]{{#2}^{#1}}

\usepackage{authblk}

\title{satsuma: Structure-based Symmetry Breaking in SAT}
\author[1]{Markus Anders}
\author[1]{Sofia Brenner}
\author[2]{Gaurav Rattan}

\affil[1]{TU Darmstadt}
\affil[2]{University of Twente}
\affil[ ]{\small\textit{\{anders,brenner\}@mathematik.tu-darmstadt.de, g.rattan@utwente.nl}}

\usetikzlibrary{arrows.meta}

\usepackage{xspace}

\newsavebox{\fminibox}
\newlength{\fminilength}

\begin{document}
\maketitle

\begin{abstract}
Symmetry reduction is crucial for solving many interesting SAT instances in practice.
Numerous approaches have been proposed, which try to strike a balance between symmetry reduction and computational overhead.
Arguably the most readily applicable method is the computation of static symmetry breaking constraints:
a constraint restricting the search-space to non-symmetrical solutions is added to a given SAT instance.
A distinct advantage of static symmetry breaking is that the SAT solver itself is not modified.
A disadvantage is that the strength of symmetry reduction is usually limited.
In order to boost symmetry reduction, the state-of-the-art tool \textsc{BreakID} [Devriendt et.~al] pioneered the identification and tailored breaking of a particular substructure of symmetries, the so-called row interchangeability groups. 

In this paper, we propose a new symmetry breaking tool called \textsc{satsuma}.
The core principle of our tool is to exploit more diverse but frequently occurring symmetry structures.
This is enabled by new practical detection algorithms for row interchangeability, row-column symmetry, Johnson symmetry, and various combinations.
Based on the resulting structural description, we then produce symmetry breaking constraints.
We compare this new approach to \textsc{BreakID} on a range of instance families exhibiting symmetry. 
Our benchmarks suggest improved symmetry reduction in the presence of Johnson symmetry and comparable performance in the presence of row-column symmetry.
Moreover, our implementation runs significantly faster, even though it identifies more diverse structures.
\end{abstract}

%%%%%%%%%%%%%%%%%%%%%%%%%%%%%%%%%%%%%%%%%%%%%%%%%%%%%%%%%%%%%%%%%%%%%%%%%%%%%%%
\section{Introduction} \label{sec:intro}
%\markus{Johnson group vs. Johnson action vs. Johnson symmetries}
%\markus{row interchangeability vs. row symmetry}
%\sofia{hyphenation in row interchangeability, row symmetry etc.}
%\sofia{consistency with $\coloneqq$ vs $=$ for definitions}
%\markus{consider a visual distiction for individualized vertices}
%\markus{remove "candidate orbit"}
%\sofia{consider citation order package}
%\subsection{Motivation}
Symmetries are present in many interesting SAT instances, ranging from hard combinatorial problems to circuit design. Making use of symmetry is paramount in order to efficiently solve many of these instances.
Practical approaches for symmetry reduction must always strike a balance between the computational overhead incurred and the strength of the symmetry reduction.
Two decades of research have led to many approaches to tackle this problem \cite{DBLP:conf/kr/CrawfordGLR96, DBLP:conf/dac/AloulMS03, DBLP:conf/sat/Devriendt0BD16, DBLP:journals/constraints/Sabharwal09, DBLP:conf/ictai/DevriendtBCDM12, DBLP:conf/sat/Devriendt0B17, DBLP:conf/tacas/MetinBCK18, DBLP:journals/jsc/JunttilaKKK20, DBLP:conf/cp/KirchwegerS21}. 
At one end of the spectrum, isomorph-free generation techniques \cite{DBLP:journals/jsc/JunttilaKKK20, DBLP:conf/cp/KirchwegerS21} apply sophisticated algorithms in conjunction with the solver, such that a solver only explores asymmetric branches of the search. While these techniques are successful in solving hard combinatorial instances (e.g., \cite{DBLP:conf/sat/KirchwegerSS22}), this comes at the price of substantial overhead:
both in terms of computational cost as well as interfering with the other strategies employed by solvers.
Hence, one must be sure that the symmetry reduction is worth the additional overhead.
It therefore seems impractical to turn these techniques ``on-by-default''.

Arguably at the other end are tools producing \emph{static symmetry breaking constraints} \cite{DBLP:conf/kr/CrawfordGLR96, DBLP:conf/dac/AloulMS03, DBLP:conf/sat/Devriendt0BD16}.
These tools add additional clauses and variables to a given instance, with the aim of reducing the number of symmetric branches explored by the solver. 
While the symmetry reduction is usually not as strong as for dynamic techniques, such constraints can be computed comparatively cheaply. 
More importantly, a distinct advantage of static symmetry breaking constraints is that the SAT solver itself is not modified, and hence there is a complete separation of concerns.
State-of-the-art static symmetry breaking tools are successfully applied as an ``on-by-default'' technique \cite{breakidkissatcomp,DBLP:conf/sat/Devriendt0BD16}.
Static symmetry breaking is not only used in SAT, but also in various other areas of constraint programming \cite{DBLP:reference/fai/GentPP06, DBLP:journals/mpc/PfetschR19, DBLP:journals/corr/Devriendt016, DBLP:conf/ijcai/AudemardJS07}.

In order to improve symmetry reduction, a rather recent development in static symmetry breaking is to detect and make use of so-called \emph{row interchangeability} subgroups \cite{DBLP:conf/sat/Devriendt0BD16, DBLP:journals/mpc/PfetschR19}.
In SAT, this feature was introduced by the state-of-the-art symmetry breaking tool \textsc{BreakID} \cite{DBLP:conf/sat/Devriendt0BD16}, but it is also used in symmetry breaking in mixed integer programming (MIP) \cite{DBLP:journals/mpc/PfetschR19}.
Row interchangeability groups stem from a natural modeling of the variables as a matrix in which all \emph{rows} are interchangeable by a symmetry.
The idea is to first identify these row interchangeability groups, and then produce tailored symmetry breaking constraints. 
The current generation of tools identifies row interchangeability by hoping for and exploiting a particular structure in the generators of the symmetry group.
However, the method is not guaranteed to work and sometimes incurs significant overhead \cite{DBLP:conf/sat/AndersSS23}.
Despite this, the gain in symmetry reduction seems to be worth the trade-off \cite{DBLP:conf/sat/Devriendt0BD16}.

In the realm of constraint programming, symmetry breaking constraints for more structures have been considered: for example, \emph{row-column symmetry} \cite{DBLP:conf/cp/FlenerFHKMPW02} is a natural extension of row interchangeability, where both the rows and columns are interchangeable.
These symmetries are common in combinatorics, scheduling, or assignment problems \cite{matrixmodelling, DBLP:conf/cp/FlenerFHKMPW02}, such as the well-known pigeonhole principle.
While traditional complete symmetry breaking constraints are unlikely to be efficiently computable for these structures \cite{DBLP:conf/cp/FlenerFHKMPW02}, different practical constraints are well-studied in the literature for \emph{manual} breaking of symmetries \cite{DBLP:conf/cp/FlenerFHKMPW02, DBLP:conf/cp/KatsirelosNW10}.
Another area in which symmetry breaking has been studied in detail is graph generation \cite{DBLP:conf/cp/CodishGIS16, DBLP:journals/constraints/CodishMPS19, DBLP:conf/cp/KirchwegerS21}.
A typical problem in this area is to determine the existence of a graph with a specific property.
Symmetries in these problems often simply correspond to \emph{isomorphic graphs}.
Even though this is rarely mentioned explicitly, such symmetries can be described by so-called \emph{Johnson groups} \cite{DBLP:conf/stoc/Babai16, permNC}.

In \emph{automated} symmetry breaking, making use of such results requires us to \emph{identify} the appropriate structures first. 
However, generalizing the existing identification strategies of contemporary tools to more elaborate structures seems elusive. 

\subsection{Contribution}
We present a new algorithm for the generation of symmetry breaking constraints, and a prototype implementation called \textsc{satsuma}. 
Our goal is to explore whether the approach of ``identifying and exploiting specific group structures'' can be pushed further.

\textbf{Techniques.}
We place the identification of specific symmetry groups at the very heart of \textsc{satsuma}.
The approach is enabled by our main contribution, a new class of practical detection algorithms.
In particular, we provide algorithms identifying row interchangeability (Section~\ref{sec:row}), row-column symmetry (Section~\ref{sec:rowcolumn}), and Johnson symmetry (Section~\ref{sec:johnson}).
Furthermore, we detect certain \emph{combinations} of the above groups, as well as groups which are \emph{similar} to the above groups, building essentially a \emph{structural description} of the group.
Symmetry breaking constraints are then chosen based on the type of detected structure: for each detected structure, we determine a set of carefully chosen symmetries, for which conventional symmetry breaking constraints are produced.

Our detection algorithms are all based on the highly efficient \emph{individualization-refinement} framework, as is commonly used in practical graph isomorphism algorithms \cite{DBLP:journals/jsc/McKayP14}.
Our detection algorithms are all heuristics, in that identification of a particular group cannot be guaranteed.
However, the success of the heuristics provably depends only on a well-studied graph property (see Section~\ref{sec:preliminaries}).

These algorithms can be applied \emph{without} computing the symmetries of the formula first: they are purely graph-based.
We exploit this by first running our tailored detection algorithms, and then only apply general-purpose symmetry detection on parts not yet identified.
In order to handle this remainder, \textsc{satsuma} reimplements parts of \textsc{BreakID}. 
Essentially, our new approach acts as a preprocessor for existing techniques.

\textbf{Benchmarks.}
We compare \textsc{satsuma} and \textsc{BreakID} on a range of well-established SAT instance families exhibiting symmetry.
In our benchmarks, we observe that our new structure-based implementation
\begin{enumerate}
	\item leads to improved SAT solver performance on instances with Johnson symmetries,
	\item comparable SAT solver performance on instances exhibiting predominantly row interchangeability or row-column symmetry,
	\item and incurs less computational overhead on all tested benchmark families (we observe better asymptotic scaling of \textsc{satsuma} on some benchmark families).
\end{enumerate}
When \textsc{satsuma} detects a structure, our new approach seems to be a win-win: it yields lower computational overhead, and the resulting speed-up for SAT solvers is comparable or better.

\section{Preliminaries} \label{sec:preliminaries}
\subsection{Satisfiability and Symmetry}
\textbf{SAT.}
In this paper, a SAT formula $F$ in \emph{conjunctive normal form} (CNF) is denoted with
$$F = \{\{l_{1,1}, \cdots{}, l_{1,k_1}\}, \cdots{}, \{l_{m,1}, \cdots{}, l_{m,k_m}\}\}.$$
Each element $C \in F$ is called a \emph{clause}, whereas a clause itself consists of a set of \emph{literals}. 
A literal is either a variable $v$ or its negation $\neg v$.
We write $\Var(F) \coloneqq \{v_1, \dots{}, v_n\}$ for the set of \emph{variables} of $F$ and use $\Lit(F)$ for its literals.

A symmetry, or \emph{automorphism}, of $F$ is a permutation $\varphi \colon \Lit(F) \to \Lit(F)$ satisfying the following properties.
First, it maps $F$ to itself, i.e., $F^\varphi = F$, where $F^\varphi$ means applying~$\varphi$ element-wise to the literals in each clause. 
Second, for all $l \in \Lit(F)$ it holds that $\neg \varphi(l) = \varphi(\neg l)$.
We define the \emph{support} of $\varphi$ as $\supp(\varphi) = \{l \in \Lit(F) \colon l^\varphi \neq l\}$, i.e., the set of all literals moved by $\varphi$. 
The set of all symmetries of~$F$ is $\Aut(F)$.
We can efficiently \emph{test} if a permutation $\varphi$ is an automorphism of $F$: for each clause $C$, we check whether $C^\varphi \in F$ holds. 

An \emph{assignment} of $F$ is a function $\theta \colon \Var(F) \to \{0, 1\}$.
We define the evaluation of $F$ under $\theta$ in the usual way, i.e., either $F[\theta] = 1$ or $F[\theta] = 0$ holds.
A formula $F$ is satisfiable if there exists an assignment $\theta$ with $F[\theta] = 1$, and unsatisfiable otherwise.
Given an assignment~$\theta$ of~$F$ and an automorphism $\varphi \in \Aut(F)$, we define 
$\theta^\varphi(v) \coloneqq \theta(v')$ if $\varphi(v) = v' \text{ for } v' \in \Var(F)$ and $\theta^{\varphi}(v) \coloneqq \neg \theta(v')$ if $\varphi(v) = \neg v'$ for $v' \in \Var(F)$,
where naturally $\neg 0 = 1$ and $\neg 1 = 0$.
It follows readily that for $\varphi \in \Aut(F)$, we have $F[\theta] = F^\varphi[\theta^\varphi] = F[\theta^\varphi]$.

\textbf{Symmetry Breaking Constraints.}
All symmetry breaking constraints in this paper are \emph{lex-leader} constraints. 
Let $\prec$ denote a total order of $\Var(F)$.
We order an assignment~$\theta$ according to $\prec$, yielding a $\{0, 1\}$-string.
We can then order assignments $\theta, \theta'$ of $F$ lexicographically by comparing their corresponding strings, denoted by $\prec_\text{lex}$. 
Given an automorphism $\varphi$ of $F$, it suffices to evaluate $F$ on those assignments $\theta$ for which $\theta^\varphi \preceq_\text{lex} \theta$ holds, since $F[\theta^\varphi] = F[\theta]$.
In particular, we may add a \emph{lex-leader constraint} $\LL^\prec_\varphi$ to~$F$, which ensures that $\theta^\varphi \preceq_\text{lex} \theta$ holds.
It is easy to see that $F$ is satisfiable, if and only if $F \bigwedge_{\varphi \in \Aut(F)} \LL^\prec_\varphi$ is satisfiable \cite{DBLP:series/faia/Sakallah21}.
Lex-leader constraints can be efficiently encoded as a CNF formula, and different encodings have been studied in detail \cite{DBLP:conf/dac/AloulMS03,DBLP:conf/sat/Devriendt0BD16}.
The practical encoding we use is reverse-engineered from \textsc{BreakID}, and is described in \cite{DBLP:conf/sat/Devriendt0BD16}. Having detected structures of symmetries, \textsc{satsuma} attempts to determine a favorable variable order and a set of automorphisms for which lex-leader constraints are constructed (see Section~\ref{sec:implementation}).

\subsection{Graphs and Symmetry}
\textbf{Graphs.}
An undirected graph $G = (V, E)$ consists of a vertex set $V$ and an edge relation $E \subseteq {V \choose 2}$.  
We refer to the set of vertices of $G$ as $V(G)$, and to the set of edges as $E(G)$.
A \emph{vertex coloring} of $G$ is a mapping $\pi : V(G) \to [k]$ to colors in $[k]$ for some $k \in \mathbb{N}$.
We call $(G, \pi)$ a vertex-colored graph. The \emph{color class} of a color $c$ consists of all vertices of $G$ with color $c$. 
The color classes form a partition of $V(G)$, the \emph{color partition} corresponding to~$\pi$.

A bijection $\varphi : V(G) \to V(G)$ is called an \emph{automorphism} of $(G, \pi)$, whenever $(G, \pi)^\varphi = (G^\varphi, \pi^\varphi) = (G, \pi)$ holds. 
Here, $G^\varphi$ denotes the graph with vertex set $V(G)$ and edges $\{u^\varphi, v^\varphi\}$ whenever $\{u,v\}$ is an edge of $G$ (where $v^\varphi$ simply denotes the image of $v$ under $\varphi$). 
The coloring $\pi^{\varphi}$ is given by $\pi^{\varphi}(v) = \pi(v^\varphi)$ for every $v \in V(G)$. 
The set of all automorphisms of $(G,\pi)$ is denoted by $\Aut(G, \pi)$.

For a given CNF formula $F$, we define the \emph{model graph} $G(F) = (G, \pi)$ as follows.
The vertex set consists of the literals and clauses of $F$.
There are edges connecting the literals of a common variable to each other.
Clauses are connected to the literals they contain.
Formally, let $E \coloneqq \{\{v, \neg v\} \colon v \in \Var(F)\} \cup \{\{C, l\} \colon l \in C,\, C \in F \}.$
Define a coloring $\pi$ by setting $\pi(l) \coloneqq 0$ for all literals $l \in \Lit(F)$ and $\pi(C) \coloneqq 1$ for all clauses $C \in F$.
It is well-known that the automorphisms of $G(F)$ restricted to $\Lit(F)$ are precisely the automorphisms of $F$ \cite{DBLP:series/faia/Sakallah21}.

\textbf{Permutation Groups.}
We recall some notions of permutation group theory. 
A detailed account can be found in \cite{seress_2003}.
Let $\Omega$ be a nonempty finite set. Let $\Sym(\Omega)$ denote the \emph{symmetric group} on $\Omega$, i.e., the group of permutations of~$\Omega$, and set $\Sym(n) \coloneqq \Sym([n])$. 
A \emph{permutation group} is a subgroup~$\Gamma$ of $\Sym(\Omega)$, denoted by $\Gamma \leq \Sym(\Omega)$. 
We also say that~$\Gamma$ \emph{acts on}~$\Omega$. 
For $g \in \Gamma$ and $\omega \in \Omega$, we write $\omega^g$ for the image of $\omega$ under $g$ and $\Orbi{\Gamma}{\omega} = \{\omega^g \colon g \in \Gamma\}$ for the \emph{orbit} of $\omega$ under $\Gamma$. In other words, $\Orbi{\Gamma}{\omega}$ consist of all points in $\Omega$ that can be reached from $\omega$ by applying elements of $\Gamma$. The partition of $\Omega$ into the orbits of $\Gamma$ is called the \emph{orbit partition}. 
For $\omega \in \Omega$, let $\Gamma_\omega \coloneqq \{g \in \Gamma \colon \omega^g = \omega \}$ denote the \emph{stabilizer} of $\omega$ in $\Gamma$. In other words, $\Gamma_\omega$ consists of those elements in $\Gamma$ that map $\omega$ to itself. The \emph{direct product} of permutation groups $\Gamma_1$ and $\Gamma_2$ which act on domains $\Omega_1$ and $\Omega_2$, respectively, is the Cartesian product $\Gamma_1 \times \Gamma_2$, endowed with a component-wise multiplication. It naturally acts component-wise on $\Omega_1 \times \Omega_2$.
\tikzset{cross/.style={cross out, draw=black, minimum size=2*(#1-\pgflinewidth), inner sep=0pt, outer sep=0pt,thick},
%default radius will be 1pt. 
cross/.default={2pt}}

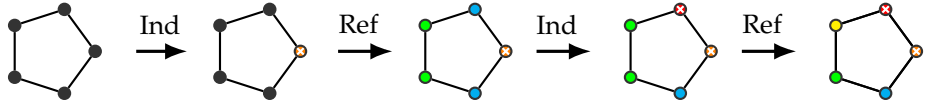
\begin{figure}
	\centering
		\begin{tikzpicture}[yscale=1,scale=0.9]
			\begin{scope}
			\foreach \a in {1,2,...,5}{
				\draw (\a*360/5: 0.65cm) node[place] (g\a) {};
			}

			\foreach \x in {1,...,5}{
				\pgfmathtruncatemacro{\xx}{\x-1}
				\pgfmathtruncatemacro{\yy}{int(Mod(\xx+1,5))}
				\pgfmathtruncatemacro{\y}{\yy+1}
				\draw [thick,draw=black] (g\x) to (g\y);
			}
			\end{scope}
			\draw [very thick, -Latex](1.25,0) -- (2,0);
			\node (i) at (1.6,0.4) {{\small Ind}};
			\begin{scope}[xshift=3cm]
				\foreach \a in {1,2,...,5}{
					\ifthenelse{\a = 5}{
						\draw (\a*360/5: 0.65cm) node[place,fill=orange] (g\a) {};
						\draw (\a*360/5: 0.65cm) node[cross,draw=white] (gc\a) {};
					}{
						\draw (\a*360/5: 0.65cm) node[place] (g\a) {};
					}
				}
	
				\foreach \x in {1,...,5}{
					\pgfmathtruncatemacro{\xx}{\x-1}
					\pgfmathtruncatemacro{\yy}{int(Mod(\xx+1,5))}
					\pgfmathtruncatemacro{\y}{\yy+1}
					\draw [thick,draw=black] (g\x) to (g\y);
				}
			\end{scope}
			\draw [very thick, -Latex](4.2,0) -- (5,0);
			\node at (4.5,0.4) {\small Ref};
			\begin{scope}[xshift=6cm]
				\foreach \a in {1,2,...,5}{
					\ifthenelse{\a = 5}{
						\draw (\a*360/5: 0.65cm) node[place,fill=orange] (g\a) {};
						\draw (\a*360/5: 0.65cm) node[cross,draw=white] (gc\a) {};
					}{
						\ifthenelse{\a = 4 \OR \a = 1}{
							\draw (\a*360/5: 0.65cm) node[place,fill=cyan] (g\a) {};
						}{
							\draw (\a*360/5: 0.65cm) node[place,fill=green] (g\a) {};
						}
					}
				}
	
				\foreach \x in {1,...,5}{
					\pgfmathtruncatemacro{\xx}{\x-1}
					\pgfmathtruncatemacro{\yy}{int(Mod(\xx+1,5))}
					\pgfmathtruncatemacro{\y}{\yy+1}
					\draw [thick,draw=black] (g\x) to (g\y);
				}
			\end{scope}
			
			\draw [very thick, -Latex](7.1,0) -- (7.9,0);
			\node at (7.4,0.4) {\small Ind};
			\begin{scope}[xshift=9cm]
			\foreach \a in {1,2,...,5}{
				\ifthenelse{\a = 5}{
					\draw (\a*360/5: 0.65cm) node[place,fill=orange] (g\a) {};
					\draw (\a*360/5: 0.65cm) node[cross,draw=white] (gc\a) {};
				}
			{
				\ifthenelse{\a = 1}{
					\draw (\a*360/5: 0.65cm) node[place,fill=red] (g\a) {};
					\draw (\a*360/5: 0.65cm) node[cross,draw=white] (gc\a) {};
				}{
				
					\ifthenelse{\a = 4 \OR \a = 1}{
						\draw (\a*360/5: 0.65cm) node[place,fill=cyan] (g\a) {};
					}{
						\draw (\a*360/5: 0.65cm) node[place,fill=green] (g\a) {};
					}
				}}
			}
			
			\foreach \x in {1,...,5}{
				\pgfmathtruncatemacro{\xx}{\x-1}
				\pgfmathtruncatemacro{\yy}{int(Mod(\xx+1,5))}
				\pgfmathtruncatemacro{\y}{\yy+1}
				\draw [thick,draw=black] (g\x) to (g\y);
			}
			\end{scope}
			\draw [very thick, -Latex](10.1,0) -- (10.9,0);
			\node at (10.4,0.4) {\small Ref};
			\begin{scope}[xshift=12cm]

			\draw (1*360/5: 0.65cm) node[place,fill=red] (g1) {};
			\draw (1*360/5: 0.65cm) node[cross,draw=white] (gc1) {};
			\draw (2*360/5: 0.65cm) node[place,fill=yellow] (g2) {};
			\draw (3*360/5: 0.65cm) node[place,fill=green] (g3) {};
			\draw (4*360/5: 0.65cm) node[place,fill=cyan] (g4) {};
			\draw (5*360/5: 0.65cm) node[place,fill=orange] (g5) {};
			\draw (5*360/5: 0.65cm) node[cross,draw=white] (gc5) {};
			
			\foreach \x in {1,...,5}{
			\pgfmathtruncatemacro{\xx}{\x-1}
			\pgfmathtruncatemacro{\yy}{int(Mod(\xx+1,5))}
			\pgfmathtruncatemacro{\y}{\yy+1}
			\draw [thick,draw=black] (g\x) to (g\y);
			}
			
			\foreach \x in {1,...,5}{
				\pgfmathtruncatemacro{\xx}{\x-1}
				\pgfmathtruncatemacro{\yy}{int(Mod(\xx+1,5))}
				\pgfmathtruncatemacro{\y}{\yy+1}
				\draw [thick,draw=black] (g\x) to (g\y);
			}
			\end{scope}
			
		\end{tikzpicture}
		\caption{An illustration of the $\IR$ process. Individualization steps break symmetries or similarities (nodes marked with a cross are individualized). Refinement steps propagate this information.}
		\label{fig:ir}
\end{figure}

%\begin{figure}
%	\centering
%	\begin{tikzpicture}[yscale=1,scale=0.8]
%	\begin{scope}
%	\foreach \a in {1,2,...,5}{
%		\draw (\a*360/5: 1cm) node[place] (g\a) {};
%	}
%	
%	\foreach \x in {1,...,5}{
%		\pgfmathtruncatemacro{\xx}{\x-1}
%		\pgfmathtruncatemacro{\yy}{int(Mod(\xx+1,5))}
%		\pgfmathtruncatemacro{\y}{\yy+1}
%		\draw [thick,draw=black] (g\x) to (g\y);
%	}
%	\end{scope}
%	\draw [very thick, -Latex](2,0) -- (3,0);
%	\begin{scope}[xshift=5cm]
%	\foreach \a in {1,2,...,5}{
%		\ifthenelse{\a = 5}{
%			\draw (\a*360/5: 1cm) node[place,fill=orange] (g\a) {};
%		}{
%			\draw (\a*360/5: 1cm) node[place] (g\a) {};
%		}
%	}
%	
%	\foreach \x in {1,...,5}{
%		\pgfmathtruncatemacro{\xx}{\x-1}
%		\pgfmathtruncatemacro{\yy}{int(Mod(\xx+1,5))}
%		\pgfmathtruncatemacro{\y}{\yy+1}
%		\draw [thick,draw=black] (g\x) to (g\y);
%	}
%	\end{scope}
%	\draw [very thick, -Latex](7,0) -- (8,0);
%	\node at (7.5,0.5) {$\IR$};
%	\begin{scope}[xshift=10cm]
%	\foreach \a in {1,2,...,5}{
%		\ifthenelse{\a = 5}{
%			\draw (\a*360/5: 1cm) node[place,fill=orange] (g\a) {};
%		}{
%			\ifthenelse{\a = 4 \OR \a = 1}{
%				\draw (\a*360/5: 1cm) node[place,fill=cyan] (g\a) {};
%			}{
%				\draw (\a*360/5: 1cm) node[place,fill=gray] (g\a) {};
%			}
%		}
%	}
%	
%	\foreach \x in {1,...,5}{
%		\pgfmathtruncatemacro{\xx}{\x-1}
%		\pgfmathtruncatemacro{\yy}{int(Mod(\xx+1,5))}
%		\pgfmathtruncatemacro{\y}{\yy+1}
%		\draw [thick,draw=black] (g\x) to (g\y);
%	}
%	\end{scope}
%	\end{tikzpicture}
%	\caption{Individualizing a vertex $v$ (orange vertex) in $\IR$, stabilizes some of the symmetry of the graph.}
%	\label{fig:ir}
%\end{figure}

\textbf{Individualization-Refinement.} 
A central ingredient in our algorithms is the so-called \emph{individualization-refinement} (IR) paradigm.
The IR paradigm is the central technique in all state-of-the-art symmetry detection algorithms \cite{DBLP:journals/jsc/McKayP14,DBLP:conf/dac/DargaLSM04,DBLP:conf/tapas/JunttilaK11,DBLP:conf/esa/AndersS21}, and highly engineered implementations are available.
The paradigm mainly consists of the \emph{individualization} technique, paired with the so-called \emph{color refinement algorithm}.
In this paragraph, we focus on a high-level explanation of the routine. 
A detailed account can be found in \cite{DBLP:journals/jsc/McKayP14}.

The central idea of IR (see Figure~\ref{fig:ir} for an illustration) is the following: given a vertex-colored graph~$(G,\pi)$ and a vertex $v \in V(G)$, the vertex $v$ is \emph{individualized}. 
Basically this means that it obtains a new color. 
The routine then proceeds with a so-called \emph{color refinement}: in each step, every vertex of $G$ obtains a new color, based on its former color together with the colors of its neighbors in $G$. 
This recoloring procedure is repeated until the corresponding color partition stabilizes. 
The final coloring $\pi'$ is then returned. 
We use the notation $\pi' = \IR((G, \pi), v)$ to denote this process.
The call $\IR((G, \pi), v)$ can be computed in time $\mathcal{O}(|E(G)|\log|V(G)|)$ (see \cite{DBLP:journals/mst/BerkholzBG17}).

The coloring $\pi'$ is a \emph{refinement} of $\pi$ in the sense that vertices with the same color in~$\pi'$ already had the same color in $\pi$. In other words, a color $c$ of $\pi$ is either preserved in $\pi'$, or partitioned into several other colors $c_1, \dots, c_n$. For $i \in [n]$, we call the sets 
$\{u \in V(G) \colon \pi(u) = c,\, \pi'(u) = c_i\}$
the \emph{fragments} of $c$ in $\pi'$. The second crucial observation is that vertices in the same orbit under the stabilizer $\Aut(G, \pi)_{v}$ obtain the same color in $\pi'$. However, it is possible that the color partition of $\pi'$ is \emph{coarser} than the orbit partition in the sense that the vertices of multiple orbits might obtain the same color in $\pi'$.

Clearly, this process can be applied inductively to individualize multiple vertices. 
%Clearly, the above process can be generalized to work with multiple individualized vertices. Given a duplicate-free list of vertices $v_1, \dots{}, v_k$ of~$G$, let $\IR(G, \pi, (v_1, \dots{}, v_k))$ denote the refined coloring obtained by individualizing $v_1, \ldots, v_k$ (each vertex obtains a different color) and applying color refinement as described above. 
It is also possible to pass the empty sequence $\epsilon$ to $\IR$, i.e., to run only the color refinement procedure. Arguing as above, the resulting color partition is guaranteed to be at least as coarse as the orbit partition of $\Aut(G,\pi)$ (i.e., the stabilizer of the empty sequence). 

The next lemma summarizes the properties of $\IR$ to which we refer throughout the paper: 

\begin{lemma} \label{lem:irreq}
%	Given a graph $G$, a vertex coloring $\pi$, and a duplicate-free list of vertices $v_1, \dots{}, v_k$ of~$G$, the refined coloring $\pi' = \IR(G, \pi, (v_1, \dots{}, v_k))$ has the following properties.
	Given a vertex-colored graph $(G,\pi)$ and a vertex $v \in V(G)$, the refined coloring $\pi' = \IR((G, \pi), v)$ has the following properties.
	\begin{enumerate}
		\item The coloring $\pi'$ is a \emph{refinement} of $\pi$: for $u,w\in V(G)$ with $\pi'(u) = \pi'(w)$, we have $\pi(u) = \pi(w)$.
		\item The color partition of $\pi'$ is at least as \emph{coarse} as the orbit partition of $\Gamma = \Aut(G, \pi)_v$: vertices $u,w\in V(G)$ with 
$\pi'(u) \neq \pi'(w)$ lie in different orbits of $\Gamma$, i.e., we have $w \not\in u^{\Gamma}$.
		\item The colors of $\pi'$ are \emph{isomorphism-invariant}: for every $\varphi \in \Sym(V(G))$, it holds that $\IR((G^\varphi, \pi^\varphi), v^\varphi) = \IR((G, \pi), v)^\varphi$. 
		In particular, if $\varphi \in \Aut(G, \pi)$, then  $\IR((G, \pi), v)^\varphi = \IR((G^\varphi, \pi^\varphi), v^\varphi) = \IR((G, \pi), v^\varphi)$ holds.
	\end{enumerate}
\end{lemma}
These properties follow almost immediately from the definition of $\IR$, and we refer to \cite{DBLP:journals/jsc/McKayP14} for a treatment of the topic. 
We also mention that usually, as opposed to the description above, $\IR$ is defined for \emph{sequences} of vertices instead of single vertices.

We now recall the notion of Tinhofer graphs \cite{DBLP:journals/cc/ArvindKRV17}. In view of the second part of Lemma~\ref{lem:irreq}, these are precisely the graphs for which the two partitions coincide. 
\begin{definition}[Tinhofer Graph \cite{DBLP:journals/cc/ArvindKRV17, DBLP:journals/dam/Tinhofer91}] \label{def:tinhofer} A graph $G$ is called Tinhofer if for all $v \in V(G)$, the orbit partition of $\Gamma \coloneqq \Aut(G,\pi)_v$ coincides with the color partition of $\pi' \coloneqq \IR((G, \pi), v)$ and the same applies recursively to the colored graph $(G,\pi')$ (this corresponds to individualizing multiple vertices of $G$). Formally, the first property means that for all $u,w \in V(G)$, we have $w \in u^\Gamma$ if and only if $\pi'(u) = \pi'(w)$. 
\end{definition}
In particular, $\IR$ works well on Tinhofer graphs: practical graph isomorphism solvers are guaranteed to not require any backtracking.
%Empirically, graphs stemming from SAT rarely cause practical graph isomorphism solvers to backtrack. \markus{citation difficult}

\subsection{Symmetry Structures in SAT}\label{sec:detectedsymmetries}
The idea of our tool is to detect certain symmetry structures that are subsequently exploited. 
In this section, we describe the main structures detected by the tool. The description of the detection algorithms is the subject of Section~\ref{sec:detectionalgorithms}. 

Throughout, let $F$ be a SAT formula. 
As a first step, consider the \emph{disjoint direct decomposition} of the symmetries $\Aut(F)$: this is a partition $\Lit(F) = L_1 \dcup \cdots \dcup L_k$ of $\Lit(F)$ for which there exists a decomposition $\Aut(F) = A_1 \times \dots \times A_k$ into a direct product of subgroups such that, for every $i \in [k]$, the automorphisms in $A_i$ only move the literals in $L_i$. 
A disjoint direct decomposition naturally decomposes the symmetry breaking problem, and it suffices to treat each factor separately.
In the following, we always refer to the finest such decomposition, which is clearly unique. 
We call its parts $L_1, \dots, L_k$ the \emph{disjoint direct factors} of $F$. 
Note that every disjoint direct factor is a union of orbits of $\Aut(F)$.

As factors in the disjoint direct decomposition, we detect several variants of three main kinds of symmetries, namely row symmetry, row-column symmetry, and Johnson symmetry. Let us now define these notions in the special context of CNF formulas. 

\textbf{Row Symmetry.} \emph{Row interchangeability}, or \emph{row symmetry}, naturally occurs in the context of matrix modeling \cite{matrixmodelling} and is already successfully exploited in automated symmetry breaking. 
We say that a SAT formula $F$ \emph{exhibits row symmetry} if there exists a disjoint direct factor $L \subseteq \Lit(F)$ which can be arranged in a matrix $M$ such that $\Aut(F)|_L$ acts by permuting the rows of $M$.
In addition, we require that every column of $M$ is an orbit of $\Aut(F)$. 
See Figure~\ref{fig:rowsym} for an illustration.
The colored boxes illustrate orbits, whereas dashed lines indicate vertices in the same row.
The rows can be permuted using symmetry.

We should address a technical difference between the definition above and how \textsc{BreakID} handles row symmetry: in our definition, a disjoint direct factor should \emph{only} admit the action of the row symmetry group, or a particularly defined extension (see Section~\ref{sec:detectionalgorithms}).
\textsc{BreakID} on the other hand would accept any row symmetry \emph{subgroup} that it detects (see \cite{DBLP:conf/sat/AndersSS23} for further discussion).
Hence, in practice, it may happen that \textsc{BreakID} reports row symmetry, when \textsc{satsuma} does not.
However, \textsc{satsuma} may instead identify a larger, more expressive group, such as row-column symmetry, as explained below.

Let us make a general observation regarding negation symmetry.
\begin{remark}
For an orbit~$\sigma$ of literals under $\Aut(F)$, also the set $\neg \sigma \coloneqq \{\neg v \colon v \in \sigma\}$ is an orbit of literals. Hence two cases can occur: either we have $\sigma = \neg \sigma$, or the orbits $\sigma$ and $\neg \sigma$ are disjoint. 
%In any case, $\sigma$ and $\neg \sigma$ belong to the same disjoint direct factor of $F$. \sofia{is this true? What happens if there are variables fixed by every symmetry?}
%\markus{it is true, unless both $v$ and $\neg v$ are in orbits of size $1$}
\end{remark}	
%Note that in the first case, we could simply assign the value of one variable $v \in \sigma$ and restrict ourselves to automorphisms of $F$ that fix $v$. 
%For this reason, our standard scenario is the second case. 
In order to simplify the exposition, we only consider the second scenario in the following.

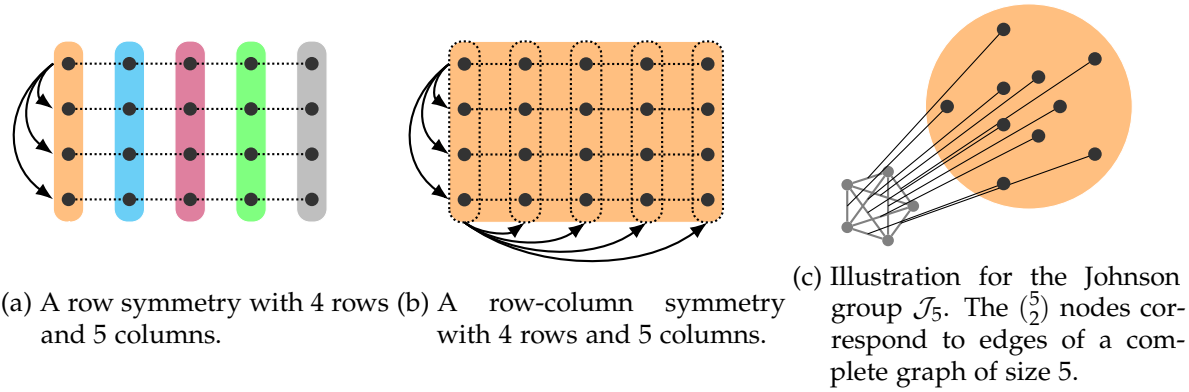
\begin{figure}
	\centering
	\begin{subfigure}[c]{0.32\textwidth}
	\begin{tikzpicture}[yscale=0.75,scale=0.8]
		\foreach \x in {1,...,5}{
			\pgfmathtruncatemacro{\colnum}{int(Mod(\x,5))}
			\ifthenelse{\colnum=0}{\def\col{gray!50}}{}
			\ifthenelse{\colnum=1}{\def\col{orange!50}}{} 
			\ifthenelse{\colnum=2}{\def\col{cyan!50}}{}
			\ifthenelse{\colnum=3}{\def\col{purple!50}}{}
			\ifthenelse{\colnum=4}{\def\col{green!50}}{}
			\draw[rounded corners=0.175cm,draw=white!0,fill=\col] (\x-0.25, 0.5) rectangle (\x+0.25, 4.5) {};
		}
		\foreach \x in {1,...,5}{
			\foreach \y in {1,...,4}{
				\node[place] (a\x\y) at (1*\x,1*\y) {};
			}
		}
		\foreach \x in {1,...,4}{
			\foreach \y in {1,...,4}{
				\pgfmathtruncatemacro{\nextx}{\x + 1}
				\draw [densely dotted, thick] (a\x\y) to (a\nextx\y);
			}
		}
		\foreach \y in {1,...,3}{
			\draw [-Latex, thick] ($(a14)-(0.25,0)$) to [bend right=45] ($(a1\y)-(0.25,0)$);
		}
		\foreach \x in {2,...,5}{
			\draw [-Latex, thick,draw=white] ($(a11)-(0,0.5)$) to [bend right=45] ($(a\x1)-(0,0.5)$);
		}
	\end{tikzpicture}
	\caption{A row symmetry with $4$ rows and $5$ columns.} \label{fig:rowsym}
	\end{subfigure}
	\begin{subfigure}[c]{0.32\textwidth}
	\begin{tikzpicture}[yscale=0.75,scale=0.8]
		\draw[rounded corners=0.175cm,draw=white!0,fill=orange!50] (1-0.25, 0.5) rectangle (5+0.25, 4.5) {};
		\foreach \x in {1,...,5}{
			\foreach \y in {1,...,4}{
				\node[place] (a\x\y) at (1*\x,1*\y) {};
			}
		}
		\foreach \x in {1,...,4}{
			\foreach \y in {1,...,4}{
				\pgfmathtruncatemacro{\nextx}{\x + 1}
				\draw [densely dotted, thick] (a\x\y) to (a\nextx\y);
			}
		}

		\foreach \x in {1,...,5}{
			\draw[rounded corners=0.175cm, densely dotted, thick, draw=black] (\x-0.25, 0.5) rectangle (\x+0.25, 4.5) {};
		}

		\foreach \y in {1,...,3}{
			\draw [-Latex, thick] ($(a14)-(0.25,0)$) to [bend right=45] ($(a1\y)-(0.25,0)$);
		}
		\foreach \x in {2,...,5}{
			\draw [-Latex, thick] ($(a11)-(0,0.5)$) to [bend right=45] ($(a\x1)-(0,0.5)$);
		}
	\end{tikzpicture}
	\caption{A row-column symmetry with $4$ rows and $5$ columns.} \label{fig:rowcolumnsym}
	\end{subfigure}
	\begin{subfigure}[c]{0.32\textwidth}
		\centering
		\begin{tikzpicture}[yscale=1,scale=0.8]

			\foreach \a in {1,2,...,5}{
				\draw (\a*360/5: 0.6cm) node[place,minimum size=1.25mm,fill=black!50,draw=black!50] (g\a) {};
			}

			\foreach \x in {1,...,5}{
				\foreach \y in {\x,...,5}{
					\draw [thick,draw=black!50] (g\x) to (g\y);
					\coordinate (ge\x\y) at ($(g\x)!0.5!(g\y)$) {};
				}
			}

			\begin{scope}[xshift=2.5cm,yshift=1.65cm]
				\draw[draw=white,fill=orange!50] (0,0) circle (1.7cm);
				\foreach \a in {1,2,...,5}{
					\coordinate (ig\a) at (\a*360/5: 1.66cm);
				}

				\foreach \x in {1,...,5}{
					\foreach \y in {\x,...,5}{
						\ifthenelse{\x=\y}{}{
						\coordinate (jd\x\y) at ($(ig\x)!0.5!(ig\y)$) {};
						}
						
					}
				}

				\foreach \x in {1,...,5}{
					\foreach \y in {\x,...,5}{
						\ifthenelse{\x=\y}{}{
						\draw [draw=black] (jd\x\y) to [] (ge\x\y);
						}
						
					}
				}

				\foreach \x in {1,...,5}{
					\foreach \y in {\x,...,5}{
						\ifthenelse{\x=\y}{}{
						\node[place] (jdn\x\y) at ($(ig\x)!0.5!(ig\y)$) {};
						}
						
					}
				}
			\end{scope}

		\end{tikzpicture}
		\caption{Illustration for the Johnson group $\mathcal{J}_5$. The $\binom{5}{2}$ nodes correspond to edges of a complete graph of size $5$.} \label{fig:johnson}
		\end{subfigure}
	\caption{Various group structures used throughout the paper. Colors indicate orbits of the group.} \label{fig:groupstructures}
\end{figure}
\begin{figure}
	\centering
	\begin{subfigure}[c]{0.45\textwidth}
		\centering
		
		\tikzstyle{jnode}=[circle,draw=black!80,thick,fill=red!30, inner sep=2pt,minimum size=4mm]			

		\tikzstyle{jedge}=[rounded rectangle,draw=black!80,thick,fill=blue!20, inner sep=0pt,minimum height=3mm, minimum width=4mm]			
		
		\begin{tikzpicture}[yscale=1,scale=0.7]

			\foreach \a in {1,2,...,5}{
				\draw (\a*360/5: 1.2cm) node[jnode,label=center:$\a$] (g\a)  {};
			}

			\foreach \x in {1,...,5}{
				\foreach \y in {\x,...,5}{
					\draw [very thick,draw=black!50] (g\x) to (g\y);
					\coordinate (ge\x\y) at ($(g\x)!0.5!(g\y)$) {};
				}
			}

			\begin{scope}[xshift=3.75cm,yshift=2.45cm]
				\foreach \a in {1,2,...,5}{
					\coordinate (ig\a) at (\a*360/5: 2cm);
				}

				\foreach \x in {1,...,5}{
					\foreach \y in {\x,...,5}{
						\ifthenelse{\x=\y}{}{
						\coordinate (jd\x\y) at ($(ig\x)!0.5!(ig\y)$) {};
						}
						
					}
				}

				\foreach \x in {1,...,5}{
					\foreach \y in {\x,...,5}{
						\ifthenelse{\x=\y}{}{
						\draw [draw=black] (jd\x\y) to [] (ge\x\y);
						}
						
					}
				}

				\foreach \x in {1,...,5}{
					\foreach \y in {\x,...,5}{
						\ifthenelse{\x=\y}{}{
						\node[jedge] (jdn\x\y) at ($(ig\x)!0.5!(ig\y)$) {\tiny $\x,\hspace{-0.05cm}\y$};
						}
						
					}
				}
			\end{scope}

			\node at (0,-1.75) {};
			
		\end{tikzpicture}
		\caption{The domain of $\mathcal{J}_5$ consists of all \\ $2$-subsets of the base set $\{1,\dots,5\}$.} \label{fig:johnson1}
		\end{subfigure}
	\hfill
	\begin{subfigure}[c]{0.5\textwidth}
	\centering
	\begin{tikzpicture}[yscale=1,scale=0.7]
	
	\tikzstyle{jnode}=[circle,draw=black!80,thick,fill=red!30, inner sep=2pt,minimum size=3mm]			
	
	\tikzstyle{jedge}=[rounded rectangle,draw=black!80,thick,fill=blue!20, inner sep=0pt,minimum height=3mm, minimum width=4mm]			
	
	\foreach \a in {1,2,...,5}{
		\draw (\a*360/5: 0.8cm) node[jnode,label=center:{\scriptsize $\a$}] (g\a)  {};
	}
	
	\foreach[evaluate={\b=int(1+\a)}] \a in {1,2,...,4}{	
		\draw[-Latex,semithick] (g\a) to (g\b) {};
	}
	\draw[-Latex, semithick] (g5) to (g1) {};

	\begin{scope}[yshift=4cm]
	
	\begin{scope}[xshift=2cm]
	
	\foreach[evaluate={\b=int(Mod(int(\a),5)+1) }] \a in {1,2,...,5}{
		\ifthenelse{\a=5}{
			\draw (\a*360/5: 0.9cm) node[jedge,label=center:{\tiny $\b,\hspace{-0.05cm}\a$}] (gl\a\b)  {};}
		{
		    \draw (\a*360/5: 0.9cm) node[jedge,label=center:{\tiny $\a,\hspace{-0.05cm}\b$}] (gl\a\b)  {};
	    }
	}

	\draw[-Latex, semithick] (gl12) to (gl23) {};
	\draw[-Latex, semithick] (gl23) to (gl34) {};
	\draw[-Latex, semithick] (gl34) to (gl45) {};
	\draw[-Latex, semithick] (gl45) to (gl51) {};
	\draw[-Latex, semithick] (gl51) to (gl12) {};

	\end{scope}
	
	\begin{scope}[xshift=5cm]

\foreach[ evaluate={\b=int(Mod(\a+2,5))} ] \a in {1,2,...,5}{
	\ifthenelse{\a<3}{
		\draw (\a*360/5: 0.9cm) node[jedge,label=center:{\tiny $\a,\hspace{-0.05cm}\b$}] (gr\a\b)  {};}
	{}
	\ifthenelse{\a=3}{
		\draw (\a*360/5: 0.9cm) node[jedge,label=center:{\tiny $\a,\hspace{-0.05cm}5$}] (gr35)  {};}
	{}
	\ifthenelse{\a>3}{
		\draw (\a*360/5: 0.9cm) node[jedge,label=center:{\tiny $\b,\hspace{-0.05cm}\a$}] (gr\a\b)  {};}
	{}
	
}

	\draw[-Latex, semithick] (gr13) to (gr24) {};
	\draw[-Latex, semithick] (gr24) to (gr35) {};
	\draw[-Latex, semithick] (gr35) to (gr41) {};
	\draw[-Latex, semithick] (gr41) to (gr52) {};
	\draw[-Latex, semithick] (gr52) to (gr13) {};

\end{scope}
\end{scope}

\draw[-Latex, thick] (1,1) -- node[midway, sloped, above] {\scriptsize induced } (2.5,2.5);
\draw[-Latex, thick] (1,1) -- node[midway, sloped, below=-0.05cm] {\scriptsize action} (2.5,2.5);

	\end{tikzpicture}
	\caption{The action of the 5-cycle on the base set and \\ its induced action on the domain of $\mathcal{J}_5$.} \label{fig:johnson2}
\end{subfigure}

\caption{An illustration of the Johnson group $\mathcal{J}_5$.   } \label{fig:johnsondetail}
\end{figure}
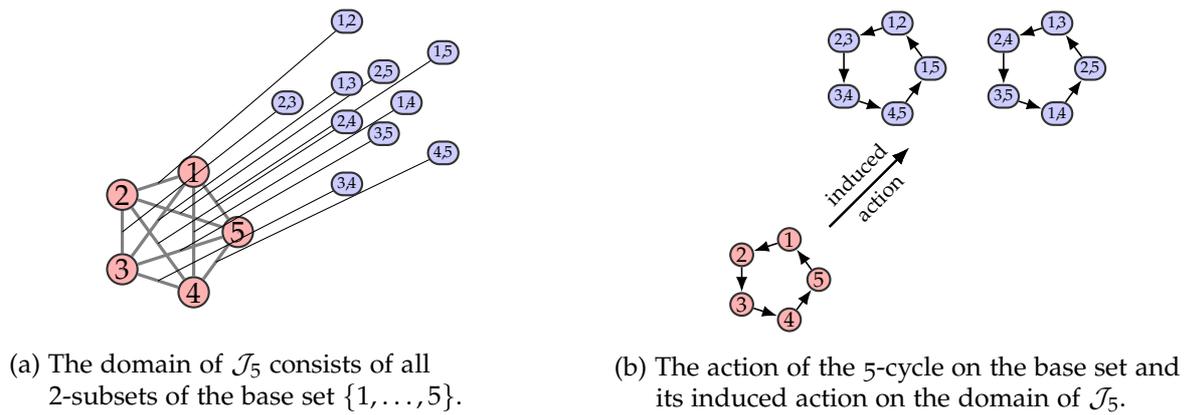

\textbf{Row-column symmetry.}  Row-column symmetries are an extension of row interchangeability. 
Row-column symmetry naturally occurs whenever both the rows and columns of a matrix of variables are interchangeable.
Examples can be found in scheduling, design, and combinatorial problems (see \cite{matrixmodelling}).

For $m, n \in \mathbb{N}$, the \emph{row-column symmetry} group is $\Gamma \coloneqq \Sym(n) \times \Sym(m)$, acting componentwise on $[n] \times [m]$. We think of $[n] \times [m]$ as an $n \times m$ matrix $M$, on which $(\sigma_1, \sigma_2) \in \Gamma$ acts by permuting the $n$ rows according to $\sigma_1$ and the $m$ columns according to~$\sigma_2$.

A SAT formula $F$ \emph{exhibits row-column symmetry} if there exists a disjoint direct factor $L \subseteq \Lit(F)$ consisting of an orbit $\sigma$ of $\Aut(F)$ and its negation $\neg \sigma$ such that the following holds: the literals in $\sigma$ can be arranged in an $n \times m$-matrix $M$ such that $\Aut(F)|_{\sigma}$ acts as a row-column symmetry group on $M$. See Figure~\ref{fig:rowcolumnsym} for an illustration. Note that the action of $\Aut(F)$ on $\sigma$ naturally extends to a row-column symmetry action on $\neg \sigma$. For this reason, our algorithm generates the matrix $M$ of the literals in $\sigma$ and extends this to $\neg \sigma$, see Section~\ref{sec:rowcolumn} for details.

\textbf{Johnson symmetry.} Johnson groups are naturally tied to the graph isomorphism problem. 
Whenever a problem asks for the existence of an undirected graph with a certain property, typically, the underlying symmetries form a Johnson group. 

%We now formally define Johnson groups. 
Observe that $\pi \in \Sym(n)$ induces a permutation on the domain ${[n] \choose 2}$ of $2$-subsets of~$[n]$, mapping $\{a_1, a_2\}$ to $\{a_1^\pi, a_2^\pi\}$.
This way, $\Sym(n)$ becomes a permutation group on a domain of size~${|n| \choose 2}$, the \emph{Johnson group} $\mathcal{J}_n$. Technically, these groups are specifically Johnson groups of arity $2$.
%Likewise, we can define an induced action of the alternating group $\Alt(k)$ on $t$-subsets of $[k]$ and denote the corresponding permutation group by $A^{(t)}_k$. The permutation groups $S_k^{(t)}$ and $A_k^{(t)}$ are called \emph{Johnson groups}, and 
The corresponding action is called a \emph{Johnson action}.  

We now define Johnson symmetries for SAT formulas. Intuitively, the variables the formula correspond to the ``edges'' (i.e., sets of two vertices) of a complete graph. There is a symmetric action on the ``vertices'' of this underlying graph and the variables of the formula (``edges'') are permuted accordingly.
See Figure~\ref{fig:johnson} and Figure~\ref{fig:johnsondetail} for an illustration. Formally, a SAT formula $F$ \emph{exhibits a Johnson symmetry} if the following holds: there exists a disjoint direct factor $L \subseteq \Lit(F)$ consisting of an orbit $\sigma$ of $\Aut(F)$ and its negation $\neg \sigma$ such that the literals in $\sigma$ can be relabeled as $x_{\{i,j\}}$ for all $\{i,j\} \in \binom{[n]}{2}$ and $\Aut(F)|_\sigma$ acts as the Johnson group $\mathcal{J}_n$ (by permuting the index sets). 
%Formally, there exists a bijection $b \colon \sigma \to \binom{[n]}{2}$ which is compatible with automorphisms in the following sense: for every $\alpha \in \Aut(F)$, there exists $\pi_\alpha \in \Sym(n)$ such that for all $v,w \in \sigma$, we have $\alpha(v) = w$ if and only if $j_{\pi_\alpha}$ maps the 2-set $b(v)$ to $b(w)$. 
%Moreover, we assume that every such Johnson action is indeed a symmetry of $F$.
%\sofia{maybe this is too complicated}
Again, the action of $\Aut(F)$ naturally extends to~$\neg \sigma$. 

\section{Detection Algorithms} \label{sec:detectionalgorithms}
We now present our detection algorithms.
All algorithms are centered around detecting structure on the model graph $G(F)$ of a given CNF formula $F$.
Recall that $G(F)$ contains a vertex for each literal, so we may use these terms interchangeably. 
The major design principles of our algorithms are described in the following.

\textbf{Colors are Orbits.}
Our algorithms work on the assumption that the model graph $G(F)$ is Tinhofer (see Definition~\ref{def:tinhofer}).
Then we can compute orbits of stabilizers using $\IR$.
In particular, the color classes of $\pi = \IR(G(F), \epsilon)$ \emph{are} then the orbits of $\Aut(G(F))$.

\textbf{Certified Correctness.}
The input model graph might not be Tinhofer.
However, each algorithm constructs a carefully chosen set of candidate permutations, which suffices to prove the existence of a certain group action.
It is then verified that these permutations are automorphisms of the formula $F$, which ensures correctness.
In our implementation, we produce lex-leader constraints only for automorphisms verified on the original formula.

\textbf{Color-by-color.} 
All of our detection algorithms proceed color-by-color, or orbit-by-orbit: given an orbit, the algorithms stabilize a specific set of points, observing the effect on the given orbit as well as other orbits.
If an orbit exhibits a specific group action, then this effect is clearly defined, and a model of the purported structure is made.

\begin{figure}
	\centering
	\begin{subfigure}[c]{0.3\textwidth}
	\begin{tikzpicture}[yscale=0.75,scale=0.7]
		\foreach \x in {1,...,5}{
			\pgfmathtruncatemacro{\colnum}{int(Mod(\x,5))}
			\ifthenelse{\colnum=0}{\def\col{gray!50}}{}
			\ifthenelse{\colnum=1}{\def\col{orange!50}}{} 
			\ifthenelse{\colnum=2}{\def\col{cyan!50}}{}
			\ifthenelse{\colnum=3}{\def\col{purple!50}}{}
			\ifthenelse{\colnum=4}{\def\col{green!50}}{}
			\draw[rounded corners=0.175cm,draw=white!0,fill=\col] (\x-0.25, 0.5) rectangle (\x+0.25, 4.5) {};
		}
		\foreach \x in {1,...,5}{
			\foreach \y in {1,...,4}{
				\node[place] (a\x\y) at (1*\x,1*\y) {};
			}
		}
		\foreach \x in {1,...,4}{
			\foreach \y in {1,...,4}{
				\pgfmathtruncatemacro{\nextx}{\x + 1}
				\draw [densely dotted, thick] (a\x\y) to (a\nextx\y);
			}
		}
		\foreach \y in {1,...,3}{
			\draw [-Latex, thick] ($(a14)-(0.25,0)$) to [bend right=45] ($(a1\y)-(0.25,0)$);
		}
		\foreach \x in {2,...,5}{
			\draw [-Latex, thick,draw=white] ($(a11)-(0,0.5)$) to [bend right=45] ($(a\x1)-(0,0.5)$);
		}
	\end{tikzpicture}
	\caption{A row symmetry with $4$ rows and $5$ columns.}
	\label{fig:row:beginning}
	\end{subfigure}
	\hspace{0.25cm}
	\begin{subfigure}[c]{0.3\textwidth}
		\begin{tikzpicture}[yscale=0.75,scale=0.7]
			\draw[rounded corners=0.175cm,draw=white!0,fill=white!0] (1-0.25, 0.5) rectangle (1+0.25, 4.5) {};
			\foreach \x in {1,...,5}{
				\pgfmathtruncatemacro{\colnum}{int(Mod(\x,5))}
				\ifthenelse{\colnum=0}{\def\col{gray}}{}
				\ifthenelse{\colnum=1}{\def\col{orange}}{} 
				\ifthenelse{\colnum=2}{\def\col{cyan}}{}
				\ifthenelse{\colnum=3}{\def\col{purple}}{}
				\ifthenelse{\colnum=4}{\def\col{green}}{}
				\ifthenelse{\colnum=0}{\def\colx{blue}}{}
				\ifthenelse{\colnum=1}{\def\colx{brown}}{} 
				\ifthenelse{\colnum=2}{\def\colx{pink}}{}
				\ifthenelse{\colnum=3}{\def\colx{yellow}}{}
				\ifthenelse{\colnum=4}{\def\colx{green}}{}
				\draw[rounded corners=0.175cm,draw=white!0,fill=\col!50] (\x-0.25, 0.5) rectangle (\x+0.25, 3.5) {};
				\draw[rounded corners=0.175cm,draw=white!0,fill=\colx!66] (\x-0.25, 3.5+0.175) rectangle (\x+0.25, 4.5-0.175) {};
			}
			\foreach \x in {1,...,5}{
				\foreach \y in {1,...,4}{
					\node[place] (a\x\y) at (1*\x,1*\y) {};
				}
			}
			
			\node[cross,draw=white] (ac14) at (1*1,1*4) {};

			\foreach \x in {1,...,4}{
				\foreach \y in {1,...,4}{
					\pgfmathtruncatemacro{\nextx}{\x + 1}
					\draw [densely dotted, thick] (a\x\y) to (a\nextx\y);
				}
			}
			\foreach \y in {1,...,2}{
				\draw [-Latex, thick] ($(a13)-(0.25,0)$) to [bend right=45] ($(a1\y)-(0.25,0)$);
			}
			\foreach \x in {2,...,5}{
				\draw [-Latex, thick,draw=white] ($(a11)-(0,0.5)$) to [bend right=45] ($(a\x1)-(0,0.5)$);
			}
		\end{tikzpicture}
		\caption{Individualizing a vertex of a row identifies the entire row.}
		\label{fig:row:ind}
		\end{subfigure}
		\hspace{0.25cm}
		\begin{subfigure}[c]{0.3\textwidth}
			\begin{tikzpicture}[yscale=0.75,scale=0.7]
				\foreach \x in {1,...,5}{
					\pgfmathtruncatemacro{\colnum}{int(Mod(\x,5))}
					\ifthenelse{\colnum=0}{\def\col{gray!50}}{}
					\ifthenelse{\colnum=1}{\def\col{orange!50}}{} 
					\ifthenelse{\colnum=2}{\def\col{cyan!50}}{}
					\ifthenelse{\colnum=3}{\def\col{purple!50}}{}
					\ifthenelse{\colnum=4}{\def\col{green!50}}{}
					\draw[rounded corners=0.175cm,draw=white!0,fill=\col] (\x-0.25, 0.5) rectangle (\x+0.25, 4.5) {};
				}
				\foreach \x in {1,...,4}{
					\foreach \y in {1,...,4}{
						\node[place] (a\x\y) at (1*\x,1*\y) {};
					}
				}
				\foreach \y in {1,...,4}{
						\node[place] (aa5\y) at (1*5,1*\y-0.25) {};
						\node[place] (ab5\y) at (1*5,1*\y+0.25) {};
					}

				\foreach \x in {1,...,3}{
					\foreach \y in {1,...,4}{
						\pgfmathtruncatemacro{\nextx}{\x + 1}
						\draw [densely dotted, thick] (a\x\y) to (a\nextx\y);
					}
				}
				\foreach \y in {1,...,4}{
						\draw [densely dotted, thick] (a4\y) to (aa5\y);
						\draw [densely dotted, thick] (a4\y) to (ab5\y);
					}

				\foreach \y in {1,...,3}{
					\draw [-Latex, thick] ($(a14)-(0.25,0)$) to [bend right=45] ($(a1\y)-(0.25,0)$);
				}
				\foreach \x in {2,...,5}{
					\draw [-Latex, thick,draw=white] ($(a11)-(0,0.5)$) to [bend right=45] ($(a\x1)-(0,0.5)$);
`				}
			\end{tikzpicture}
			\caption{A row symmetry on blocks of size $2$.}
			\label{fig:row:block}
			\end{subfigure}
			\caption{Illustrations of different aspects of row symmetry.} \label{fig:row}
\end{figure}
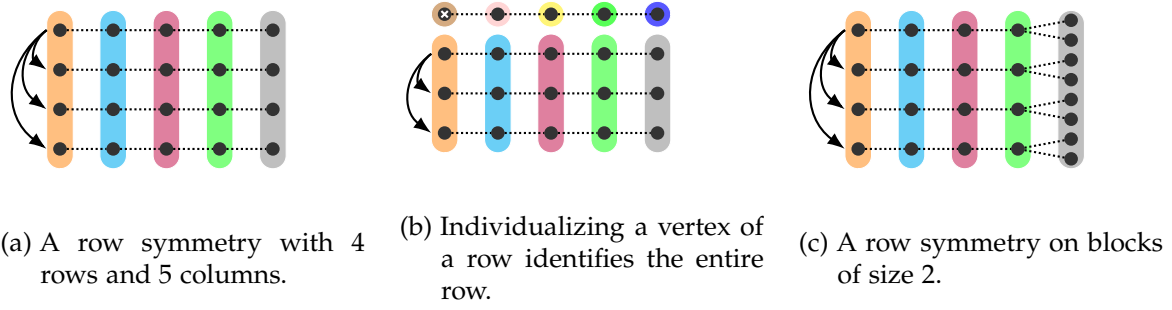

\subsection{Row Symmetry} \label{sec:row}
We describe an algorithm for row symmetry. First, we define an auxiliary function that transposes two pair-wise disjoint lists of literals of equal length: For $l \in \Lit(F)$, let
%Let $\transpose_F((l_1\dots{}l_k), (l_1'\dots{}l_k'))_F(l) : \Lit(F) \to \Lit(F)$ with 
$$\transpose_F((l_1\dots{}l_k), (l_1'\dots{}l_k'))_F(l) \coloneqq 
\begin{cases}
	l_i' &\text{ if } l = l_i \text{ with } i \in [k]\\
	l_i &\text{ if } l = l_i' \text{ with } i \in [k]\\
	l &\text{ otherwise.}
\end{cases}
$$
\SetKwFunction{Row}{DetectRowSymmetry}
\IncMargin{0.5em}
\begin{algorithm}[]
	\SetAlgoLined
	\SetAlgoNoEnd
	\caption[Refinement procedure]{Detection algorithm for row symmetry.}\label{alg:row}
	\Fn{\Row}{
		\Input{\inputb{formula $F$}\inputb{set $\sigma \subseteq \Lit(F)$ with $|\sigma| \geq 3$}}
		\Output{\outputb{matrix with row symmetry including $\sigma$, or $\bot$ if check fails}}
		$(G, \pi) \coloneqq G(F) $, $\pi' \coloneqq \IR((G, \pi), \epsilon)$\;
		
		\medskip
		\tcp{construct a candidate row for each $v \in \sigma$}
		\ForEach{$v \in \sigma$}{
			$\pi_v \coloneqq \IR((G,\pi'),v)$\;
%			individualize $v$ in $(G, \pi)$\;
%			apply color refinement\;
			%check that all fragments of $\sigma$ in $\pi'$ are singletons\;
			let $\tau$ be a list of literals that are singletons in $\pi_v$ but not in $\pi'$\;\label{alg:row:line:singleton}
			sort literals in $\tau$ according to their color in $\pi_v$\; \label{alg:row:line:ordersingleton}
			$\text{row}[v] \coloneqq \tau$\;
%			undo individualization and refinement\;
		}
		check that rows are pair-wise disjoint\;
		
		\medskip
		\tcp{verify that $M$ exhibits row symmetry}
		\ForEach{$i \in \{1 \dots{} |\sigma|-1\}$}{ \label{alg:row:check}
		$v \coloneqq \sigma[i-1]$; $v' \coloneqq \sigma[i]$\;
		check that $\transpose_F(\text{row}[v], \text{row}[v'])$ is a symmetry of $F$\;
		}
		\Return{matrix $M$ constructed from $\text{row}$}
	}
\end{algorithm}
\DecMargin{0.5em}

\textit{(Description of Algorithm~\ref{alg:row}.)} For an illustration, see Figure~\ref{fig:row:beginning}. 
The algorithm applies $\IR$ for each $v \in \sigma$ (see Figure~\ref{fig:row:ind}). 
%Then, it is checked whether the individualization of $v$ causes vertices $v'$ in other orbits to be individualized, i.e., whether fixing $v$ also fixes $v'$.
All vertices $v'$ in other orbits which are individualized in this process, i.e., which are fixed once $v$ is fixed, are added to the purported ``row'' of $v$.
%Having constructed a row for each vertex of $\sigma$, 
We then verify that every row transposition of the resulting matrix is indeed a symmetry of $F$.

\textit{(Correctness of Algorithm~\ref{alg:row}.)}
We first make the following observation for orbits of stabilizers in row interchangeability groups.
\begin{lemma} \label{lem:row:stab}
	Let $\Gamma = \Sym(n)$ be a row interchangeability group acting on $[n] \times [m]$. For every $(i,j) \in [n] \times [m]$, 
	%we have $\Gamma_{(i,j)} = \Sym([n] \setminus \{i\})$
	%In particular, $G_{(i,j)}$ is generated by all row transpositions not involving the $i$-th row. 
	the orbit of $(k,l) \in [n] \times [m]$ under the stabilizer $\Gamma_{(i,j)}$ of $(i,j)$ is given by 
	\[(k, l)^{\Gamma_{(i,j)}} = \begin{cases}
	\{(i,l)\} &\text{if }	k = i \\
	([n] \setminus \{i\}) \times \{l\} &\text{ otherwise.}
	\end{cases}
\]
\end{lemma}	
\begin{proof}
Interpreting $[n] \times [m]$ as $n \times m$-matrix $M$, recall that $\Gamma$ acts by permuting the rows of $M$. In other words, the stabilizer $\Gamma_{(i,j)}$ consisting of all row permutations that fix the $i$-th row and permute the other rows arbitrarily. Now consider the orbit of $(k,l) \in [n] \times [m]$ under the stabilizer $\Gamma_{(i,j)}$. If $k = i$, then $(k,l)$ can only be mapped to elements in the same row as $\Gamma_{(i,j)}$ fixes the $i$-th row of $M$. On the other hand, since $\Gamma$ acts by permuting the rows, every element of $M$ can only be mapped to elements in the same column, that is, $(k,l)$ must be fixed. Similarly, for $k \neq i$, the element $(k,l)$ can be mapped to all elements in the $l$-th column except for $(i,l)$. 	
\end{proof}
Next, we prove that the algorithm always returns correct symmetries of $F$ and that in case the model graph is Tinhofer, the algorithm is guaranteed to detect row interchangeability groups.
\begin{restatable}{theorem}{corrrowsym} 
	Let $F$ be a SAT formula.
	\begin{enumerate}
		\item If Algorithm~\ref{alg:row} returns a matrix $M$, every row permutation of $M$ is a symmetry of $F$.
		\item If $F$ exhibits row interchangeability with at least three rows including the input set $\sigma$ and $G(F)$ is a Tinhofer graph, Algorithm~\ref{alg:row} detects this structure and returns a corresponding matrix of literals.
	\end{enumerate}
\end{restatable}
\begin{proof}
	The first claim is guaranteed by the last part of Algorithm~\ref{alg:row} which ensures that transpositions of the rows of the returned matrix $M$ are indeed symmetries of $F$ (Line~\ref{alg:row:check}). 
	This implies that arbitrary row permutations are symmetries of $F$. 
	
	Now assume that $F$ exhibits a row symmetry with at least three rows including $\sigma$ and $G(F)$ is Tinhofer.
	We argue that the algorithm successfully detects this symmetry. 
	We remark that the orbits of $\Aut(G(F))$ restricted to the literals are precisely orbits of $\Aut(F)$.
	Let $L$ be the disjoint direct factor of $F$ containing $\sigma$ and assume that the literals in $L$ can be partitioned into a matrix $M$ that exhibits row symmetry (see Figure~\ref{fig:row:beginning}). 
	Due to the assumption that $G(F)$ is Tinhofer, if the vertex $v$ corresponding to a literal $l$ of $F$ is individualized, the resulting refined coloring consists of the orbits of $\Aut(G(F))_v$.
	In particular, due to Lemma~\ref{lem:row:stab}, the vertices in the row of $M$ are fixed and all other vertices are contained in orbits of size at least two since we have at least three rows (see Figure~\ref{fig:row:ind}). 
	Note that since we have at least three rows, $\neg l$ must be in the row of~$l$. 
	Hence after executing the loop for $v$, $\text{row}[v]$ contains precisely the vertices in the row of $v$. 
	Isomorphism-invariance of the IR routine (see Lemma~\ref{lem:irreq}) ensures that for each row, the order in which symmetrical singletons are colored will be consistent in each row (see Line~\ref{alg:row:line:ordersingleton}).
	This ensures that the rows we construct can indeed be transposed (see Line~\ref{alg:row:check} onwards), and the algorithm correctly returns a corresponding matrix.
	\end{proof}

\textbf{Recursive Row Symmetry.}
In practice, orbits often do not \emph{just} exhibit a row symmetry. 
In particular, we consider the case that an orbit of size $k$, with a natural symmetric action, is connected to another orbit of size $ck$, where the symmetric action acts on blocks of size $c$ (see Figure~\ref{fig:row:block}). 
We extend our algorithm to detect this particular case as follows: 
in Line~\ref{alg:row:line:singleton}, we add fragments of other colors instead of vertices in singletons to the row.
%we not only add vertices of singleton colors to the row, but fragments of other colors.
Let $c$ be a color of $\pi$ with a fragment $c'$ in $\pi'$.
We add the vertices $\pi'^{-1}(c')$ to the row, whenever $|\pi'^{-1}(c')||\sigma| = |\pi^{-1}(c)|$.
This means we consider vertices of $c'$, whenever there is the possibility that the color $c$ is split into $|\sigma|$ parts of size $|\pi'^{-1}(c')|$.
We call $\pi'^{-1}(c')$ a \emph{block} of its orbit.
On these blocks, we call our algorithm for row symmetry recursively.
Essentially, this enables us to detect recursive structures of row symmetry.  

\textbf{Row Symmetry in Stabilizer.}
A slight extension is that if the test for row symmetry fails, we recurse on the largest fragment from the first $\IR$ call and check whether it exhibits row symmetry.
This extension is used for the other detection algorithms as well.

\subsection{Row-Column Symmetry} \label{sec:rowcolumn}

Next, we describe a detection algorithm for row-column symmetry. 
As discussed in Section~\ref{sec:detectedsymmetries}, a disjoint direct factor exhibiting row-column symmetry consists of an orbit of literals and its negation, which is also an orbit of literals. 
We detect row-column symmetry only on one of these orbits, and expand the resulting automorphisms to the other one: For a permutation $\varphi$ of $\Lit(F)$ and all $l \in \Lit(F)$, let
%, let .
%Let $\varphi$ be a permutation of the literals.
%We define $\fix_{F}(\varphi) : \Lit(F) \to \Lit(F)$ where for each $l \in \Lit(F)$
$$ \fix_{F}(\varphi)(l) \coloneqq 
\begin{cases}
	\varphi(l) &\text{ if } l \in \supp(\varphi)\\
	\neg\varphi(\neg l) &\text{ if } \neg l \in \supp(\varphi)\\
	l &\text{ otherwise.}
\end{cases}$$

\begin{figure}
	\centering
	\begin{subfigure}[c]{0.4\textwidth}
		\centering
		\begin{tikzpicture}[yscale=0.75,scale=0.7]
			\draw[rounded corners=0.175cm,draw=white!0,fill=orange!50] (1-0.25, 0.5) rectangle (5+0.25, 4.5) {};
			\foreach \x in {1,...,5}{
				\foreach \y in {1,...,4}{
					\node[place] (a\x\y) at (1*\x,1*\y) {};
				}
			}
			\foreach \x in {1,...,4}{
				\foreach \y in {1,...,4}{
					\pgfmathtruncatemacro{\nextx}{\x + 1}
					\draw [densely dotted, thick] (a\x\y) to (a\nextx\y);
				}
			}
	
			\foreach \x in {1,...,5}{
				\draw[rounded corners=0.175cm, densely dotted, thick, draw=black] (\x-0.25, 0.5) rectangle (\x+0.25, 4.5) {};
			}
	
			\foreach \y in {1,...,3}{
				\draw [-Latex, thick] ($(a14)-(0.25,0)$) to [bend right=45] ($(a1\y)-(0.25,0)$);
			}
			\foreach \x in {2,...,5}{
				\draw [-Latex, thick] ($(a11)-(0,0.5)$) to [bend right=45] ($(a\x1)-(0,0.5)$);
			}
		\end{tikzpicture}
		\caption{A row-column symmetry with $4$ rows and $5$ columns.}
		\label{fig:rowcolumn:beginning}
	\end{subfigure}
	\hspace{0.5cm}
	\begin{subfigure}[c]{0.4\textwidth}
		\centering
		\begin{tikzpicture}[yscale=0.75,scale=0.7]
			\draw[rounded corners=0.175cm,draw=white!0,fill=white!0] (1-0.25, 0.5) rectangle (5+0.25, 4.5) {};
			\draw[rounded corners=0.175cm,draw=white!0,fill=orange!50] (2-0.25, 0.5) rectangle (5+0.25, 3.5) {};
			\draw[rounded corners=0.175cm,draw=white!0,fill=cyan!50] (2-0.25, 3.66) rectangle (5+0.25, 4.33) {};
			\draw[rounded corners=0.175cm,draw=white!0,fill=purple!50] (1-0.25, 0.5) rectangle (1+0.25, 3.5) {};
			\draw[rounded corners=0.175cm,draw=white!0,fill=gray!50] (1-0.25, 3.66) rectangle (1+0.25, 4.33) {};
			\foreach \x in {1,...,5}{
				\foreach \y in {1,...,4}{
					\node[place] (a\x\y) at (1*\x,1*\y) {};
				}
			}
			\node[cross,draw=white] (ac14) at (1*1,1*4) {};

			\foreach \x in {1,...,4}{
				\foreach \y in {1,...,4}{
					\pgfmathtruncatemacro{\nextx}{\x + 1}
					\draw [densely dotted, thick] (a\x\y) to (a\nextx\y);
				}
			}
	
			\foreach \x in {1,...,5}{
				\draw[rounded corners=0.175cm, densely dotted, thick, draw=black] (\x-0.25, 0.5) rectangle (\x+0.25, 4.5) {};
			}
	
			\foreach \y in {1,...,2}{
				\draw [-Latex, thick] ($(a13)-(0.25,0)$) to [bend right=45] ($(a1\y)-(0.25,0)$);
			}
			\foreach \x in {3,...,5}{
				\draw [-Latex, thick] ($(a21)-(0,0.5)$) to [bend right=45] ($(a\x1)-(0,0.5)$);
			}
			\draw [-Latex, thick, draw=none] ($(a11)-(0,0.5)$) to [bend right=45] ($(a51)-(0,0.5)$);
		\end{tikzpicture}
		\caption{Individualizing a vertex identifies its row and column.}
		\label{fig:rowcolumn:ind}
	\end{subfigure}
	\caption{Illustrations of different aspects of row-column symmetry.}
	\label{fig:rowcolumn}
\end{figure}
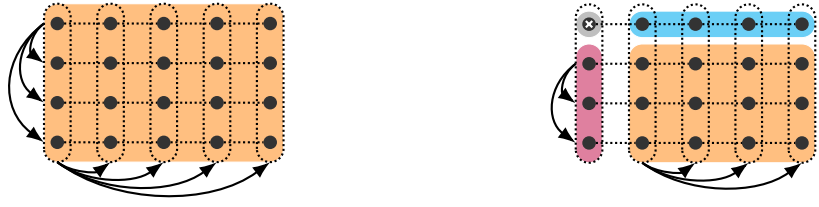

\SetKwFunction{RowColumn}{DetectRowColumnSymmetry}
\IncMargin{0.5em}
\begin{algorithm}[ht]
	\SetAlgoLined
	\SetAlgoNoEnd
	\caption[Refinement procedure]{Detection algorithm for row-column symmetry.}\label{alg:rowcolumn}
	\Fn{\RowColumn}{
		\Input{\inputb{formula $F$}\inputb{set $\sigma \subseteq \Lit(F)$ }}
		\Output{\outputb{candidate matrix $M$, or $\bot$ if check fails}}
		$(G, \pi) \coloneqq G(F) $, $\pi' \coloneqq \IR((G, \pi),\epsilon)$\;
		choose arbitrary $v \in \sigma$\; 
		$\pi_v \coloneqq \IR((G,\pi),v)$\; \label{alg:rowcolumn:ircall}
		check that $\sigma$ has $4$ fragments in $\pi_v$\;
		label fragments of $\sigma$ in $\pi_v$ not containing $v$ as $\sigma_1$, $\sigma_2$, $\sigma_3$ in increasing size\;
		
		\medskip
		\tcp{we determine ``coordinates'' in matrix relative to $v$}
		$\text{row}[v] \coloneqq \text{col}[v] = v$  \tcp*{$v$ defines a row and a column}
		\ForEach{$r \in \sigma_1$}{
			$\text{row}[r] \coloneqq v$, $\text{col}[r] \coloneqq r$ \tcp*{$r$ is in row of $v$, and defines a column}
		    $\pi_r \coloneqq \IR((G,\pi'),r)$\;
			let $\tau$ be the fragment of $\sigma$ in $\pi_r$ of size $|\sigma_2|$ not containing $v$ if exists\;
			\lForEach*{$t \in \tau$}{$\text{col}[t] \coloneqq r$\tcp*{$t$ is in column of $r$}\label{alg:rowcol:setcol}}
		}
					
		\ForEach{$c \in \sigma_2$}{
			$\text{col}[c] \coloneqq v$, $\text{row}[c] = c$ \tcp*{$c$ is in column of $v$, and defines a row}
			$\pi_c \coloneqq \IR((G,\pi'),c)$\;
			let $\tau$ be the fragment of $\sigma$ in $\pi_c$ of size $|\sigma_1|$ not containing $v$ if exists\;
			\lForEach*{$t \in \tau$}{$\text{row}[t] \coloneqq c$\tcp*{$t$ is in row of $c$}\label{alg:rowcol:setrow}}
		}
		construct matrix $M$ where $M[r, c] = v'$ with $\text{row}[v'] = r$ and $\text{col}[v'] = c$\;
		\medskip
		\tcp{verify that $M$ exhibits row-column symmetry}
		check that every vertex in $\sigma$ has a unique row and a unique column label\;
		check that distinct vertices are assigned distinct label pairs\;
		check that $M$ has pairwise disjoint rows, and pairwise disjoint columns\;
		\lForEach*{$r \in \sigma_1$}{ \label{alg:rowcol:check1}
			check that $\fix_F(\transpose_F(M[*, r], M[*, v]))$ is a symmetry of $F$\tcp*{$M[*, x]$ denotes column of $x$}
		}
		\lForEach*{$c \in \sigma_2$}{ \label{alg:rowcol:check2}
			check that $\fix_F(\transpose_F(M[c, *], M[v, *]))$ is a symmetry of $F$ \tcp*{$M[x, *]$ denotes row of $x$}
		}	
		\Return{$M$}
	}
\end{algorithm}
\DecMargin{0.5em}

\textit{(Description of Algorithm~\ref{alg:rowcolumn}.)} 
For an illustration, see Figure~\ref{fig:rowcolumn:beginning}. Given a set $\sigma \subseteq \Lit(F)$, we apply $\IR$ to a fixed vertex $v \in \sigma$ (see Figure~\ref{fig:rowcolumn:ind}). Assuming that a row-column symmetry is present, this determines a purported ``row'' $\text{row}[v]$ and ``column'' $\text{col}[v]$ of $v$. The algorithm now successively individualizes the vertices in $\text{row}[v]$ and $\text{col}[v]$. This way, every vertex in $\sigma$ is assigned a reference vertex in each of $\text{row}[v]$ and $\text{col}[v]$, determining its position in the purported matrix. We then verify that the matrix is well-defined and that every row and column transposition, expanded to $\neg \sigma$, is indeed a symmetry of $F$. 
%
%
%It is then checked whether this causes $\sigma$ to split into four orbits, one of them consisting of the singleton $v$. The other orbits are labeled $\sigma_1$, $\sigma_2$, $\sigma_3$ and sorted in increasing size. Now, one after the other, every vertex in $\sigma_1$ is individualized. If this causes a vertex to become a singleton in the resulting refinement, these vertices belong to the purported row of $w$. After that, the vertices in $\sigma_2$ are individualized, leading to purported columns of the matrix. Having constructed a row and a column for each vertex of $\sigma$, the candidate matrix is then verified by checking whether each transposition of rows and of the columns is a symmetry of $F$.
%\markus{this description is not right, things do not become singleton... needs to be aligned again to algo}

%\textit{(Runtime of Algorithm~\ref{alg:rowcolumn}.)} Analogously to the case of row interchangeability, the runtime of the algorithm is given by $\mathcal{O}(|\sigma|(|F|+(|G|\log|G|)))$. 

\textit{(Correctness of Algorithm~\ref{alg:rowcolumn}.)} 
In order to prove the correctness of Algorithm~\ref{alg:rowcolumn}, we first observe the following: 
\begin{lemma} \label{lem:rowcolumn:observe}
	Let $\Gamma = \Sym(n) \times \Sym(m)$ be a row-column symmetry group acting on $[n] \times [m]$. For every $(i,j) \in [n] \times [m]$, 
	%the stabilizer is given by $\Gamma_{(i,j)} = \Sym([n] \setminus \{i\}) \times \Sym([m] \setminus \{j\})$. 
	%In particular, $G_{(i,j)}$ is generated by all row transpositions not involving the $i$-th row and all column transpositions not involving the $j$-th column. 
	the orbit of $(k,l) \in [n] \times [m]$ under the action of $\Gamma_{(i,j)}$ is given by 
	\[ (k,l)^{\Gamma_{(i,j)}} = \begin{cases}
	\{(k,l)\} & \text{if } (k,l) = (i,j) \\
	\{i\} \times ([m] \setminus \{j\}) & \text{if } k = i, \, l \neq j \\
	([n] \setminus \{i\}) \times \{j\} &\text{if } k \neq i, \, l = j \\
	([n] \setminus \{i\}) \times ([m] \setminus \{j\}) & \text{otherwise.}
	\end{cases}	
	\]
	\end{lemma}	
	
	\begin{proof}
	We identify $[n] \times [m]$ with the entries of an $n \times m$-matrix $M$. Then $\Gamma$ acts on $M$ by permuting the rows and the columns of $M$. Let $\pi \in \Gamma$ be a permutation that fixes the entry $(i,j)$. Write $\pi = (\pi_r, \pi_c)$, where $\pi_r$ is a permutation of the rows and $\pi_c$ a permutation of the columns of $M$. Then $\pi_r$ fixes the $i$-th row and $\pi_c$ fixes the $j$-th column of $M$. On the other hand, every such element of $\Gamma$ fixes the entry $(i,j)$.  
	
	Now consider the orbit of $(k,l) \in [n] \times [m]$ under the stabilizer $\Gamma_{(i,j)}$. By definition, it consists of $(k,l)$ for $(k,l) = (i,j)$. For $k = i$ and $l \neq j$, we can map $(k,l) = (i,l)$ to all elements in the $i$-row, except for $(i,j)$. Similarly, we argue if $k \neq i$ and $l = j$. Finally, if $k \neq i$ and $l \neq j$, we can map $(k,l)$ to all vertices $(k', l')$ with $k' \neq i$ and $l' \neq j$. This shows the claim. 
	\end{proof}
	
	We prove that the algorithm always returns correct symmetries of $F$ and that in case the model graph is Tinhofer, it is guaranteed to detect row-column symmetry groups.

\begin{restatable}{theorem}{corrrowcolumnsym}
Let $F$ be a SAT formula. 
\begin{enumerate}
\item If Algorithm~\ref{alg:rowcolumn} returns a matrix $M$ of literals, every permutation of the rows or the columns of $M$, expanded to the negations of the literals in $M$, is a symmetry of $F$.
\item If $F$ exhibits a row-column symmetry with at least three rows and at least three columns including $\sigma$ and $G(F)$ is a Tinhofer graph, then Algorithm~\ref{alg:rowcolumn} detects this structure and returns a corresponding matrix representation of the literals in $\sigma$. 
\end{enumerate}
\end{restatable}
\begin{proof}
The first claim is guaranteed by the last part of Algorithm~\ref{alg:rowcolumn} which ensures that transpositions of the rows (Line~\ref{alg:rowcol:check1}) and columns (Line~\ref{alg:rowcol:check2}) of the returned matrix $M$, expanded to the corresponding negated literals, are indeed symmetries of $F$. 
By suitably composing such transpositions, we obtain that every permutation of the rows or columns of $M$ induces a symmetry of $F$ in this way.

Now assume that $G(F)$ is Tinhofer and that $F$ exhibits row-column symmetry with at least three rows and columns on $\sigma$. In other words, the literals in $\sigma$ can be arranged in a matrix $M$ on which $\Aut(F)$ acts by row and column permutations (see Figure~\ref{fig:rowcolumn:beginning}). Individualizing a fixed vertex $v \in \sigma$ causes $\sigma$ to split into four fragments according to the orbits of the stabilizer $\Aut(F)_v$: 
the singleton $\{v\}$, two fragments $\sigma_1$ and $\sigma_2$ corresponding to the remainders of the row and the column of $M$ containing $v$, and a fragment $\sigma_3$ containing the remaining vertices (see Lemma~\ref{lem:rowcolumn:observe} Figure~\ref{fig:rowcolumn:ind}).
Since we assume that $M$ has at least three rows and columns, $\sigma_1, \sigma_2, \sigma_3$ are non-singletons and $\sigma_3$ is the largest fragment. Without loss of generality, let $\sigma_1 \cup \{v\}$ be the row and $\sigma_2 \cup \{v\}$ be the column of~$v$ in~$M$. Every column of $M$ is determined by the unique element of $\sigma_1 \cup \{v\}$ that it contains (similarly for the rows). Individualizing a vertex $r \in \sigma_1$ leads a similar split of $\sigma$ into four fragments. The fragments corresponding to the row and column of $r$ can be distinguished by observing that $v$ lies in the same row, but not in the same column as $r$. For all vertices in the column of $r$, we store this information (Line~\ref{alg:rowcol:setcol}). Similarly, we proceed for the columns (Line~\ref{alg:rowcol:setrow}). 
After this procedure, every element of $\sigma$ is assigned a row and column representative in $\sigma_2$ and $\sigma_1$ respectively, which, up to a permutation of the rows and columns, allows us to recover the matrix~$M$. 
\end{proof}

\subsection{Johnson Symmetry} \label{sec:johnson}

Finally, we describe a procedure to detect Johnson actions. 
We remark that there is a classic algorithm to detect Johnson groups \cite{permNC}.
A difference to our heuristic is that we do not know the generators of the group, and instead apply techniques directly on a given graph.

Our aim is to identify the variables in the input set $\sigma$ with the 2-subsets of $[n]$, where $|\sigma| = \binom{n}{2}$. 
We thus search for a bijection $b \colon \sigma \to \binom{[n]}{2}$ such that $\Aut(F)$ acts as the Johnson group $\mathcal{J}_n$ on $\sigma$ via this bijection
%that is compatible with $\Aut(F)$ in the following sense: for every $\varphi \in \Aut(F)$, there exists $j_\varphi \in \mathcal{J}_n$ with the following property: for all $v \in \sigma$, we have $\varphi(v) = w$ precisely if $j_{\varphi}(b(v)) = b(w)$. In other words, if we identify the literals $v \in \sigma$ with 2-sets $b(v)$, the automorphism $\varphi$ induces a Johnson action 
(see Section~\ref{sec:detectedsymmetries}, Figure~\ref{fig:groupstructures}, and Figure~\ref{fig:johnsondetail}). To avoid confusion, we refer to the elements of $[n]$ as \emph{labels} and to those of $\Lit(F)$ as \emph{literals} or \emph{vertices} of $G(F)$.

\SetKwFunction{DetectJohnson}{DetectJohnson}
\IncMargin{0.5em}
\begin{algorithm}[]
	\SetAlgoLined
	\SetAlgoNoEnd
	\caption[Refinement procedure]{Detection algorithm for Johnson actions.}\label{alg:johnson}
	\Fn{\DetectJohnson}{
		\Input{\inputb{formula $F$}\inputb{set $\sigma \subseteq \Lit(F)$ }}
		\Output{\outputb{bijective labeling of $\sigma$ by 2-subsets of $[n]$, or~$\bot$ if check fails}}
		$(G, \pi) \coloneqq G(F) $, $\pi' \coloneqq \IR((G, \pi), \epsilon)$\;
		
		check that $|\sigma| \geq 28$ and $|\sigma| = \binom{n}{2}$ for some $n \in \mathbb{N}$\;
		\lForEach{vertex $v \in \sigma$}{set $\text{label}[v] = []$}
		vnr = 1\;
		
		%\tcp{compute vertex labels}
		\While{there are vertices $v \in \sigma$ with $|\text{label}[v]| \leq 1$}{
			$E_i \coloneqq E_j \coloneqq E_k \coloneqq \{\}$\;
			choose $v \in \sigma$ with $|\text{label}[v]| \leq 1$\; 
			$\pi_v \coloneqq \IR((G,\pi'),v)$\; \label{alg:johnson:firstindividualization}
			%			individualize $v$ in $(G, \pi)$\;
			%			apply color refinement to obtain coloring $\pi_v$\;
			check that number of fragments of $\sigma$ in $\pi_v$ is $3$\;
			let $\sigma_v$ be the smaller non-singleton fragment\;
			\lForEach{$x \in \sigma_v$}{\label{alg:johnson:adv}
				add $x$ to $\text{ad}[v]$
			}

			choose arbitrary $w \in \text{ad}[v]$\;
			%and individualize $w$ in $(G,\pi)$\;
			%apply color refinement to obtain coloring $\pi_w$\;
			$\pi_w \coloneqq \IR((G,\pi'),w)$\;
			check that number of fragments of $\sigma$ in $\pi_w$ is $3$\;
			let $\sigma_w$ be the smaller non-singleton fragment\;
			\lForEach{$x \in \sigma_w$}{\label{alg:johnson:adw} add $x$ to $\text{ad}[w]$}	
						
			$\pi_{v,w} \coloneqq \IR((G,\pi_v),w)$\;
			
			%			individualize $w$ in $(G, \pi_v)$\;
			%			apply color refinement to obtain coloring $\pi_{v,w}$\;
			
			let $\{y\}$ be the unique singleton fragment of $\sigma$ in $\pi_{v,w}$ different from $\{v\}$ and $\{w\}$ if existent, otherwise \Return{$\bot$}\;
			$\pi_y \coloneqq \IR((G,\pi'),y)$\;
			check that number of fragments of $\sigma$ in $\pi_y$ is $3$\;
			let $\sigma_y$ be the smaller non-singleton fragment\;
			\lForEach{$x \in \sigma_y$}{\label{alg:johnson:ady} add $x$ to $\text{ad}[y]$}
			
			add $v$ to $E_i$ and $E_j$, add $w$ to $E_j$ and $E_k$, add $y$ to $E_i$ and $E_k$\;
			
			\lForEach{$x \in \text{ad}[v] \cap \text{ad}[y]$ and $x \neq w$}{\label{alg:johnson:ei}add $x$ to $E_i$}
			
			\lForEach{$x \in \text{ad}[v] \cap \text{ad}[w]$ and $x \neq y$}{\label{alg:johnson:ej}add $x$ to $E_j$}
		
			\lForEach{$x \in \text{ad}[w] \cap \text{ad}[y]$ and $x \neq v$}{\textbf{\label{alg:johnson:ek}}add $x$ to $E_k$}

			\ForEach{$E \in \{E_i, E_j, E_k\}$}{\label{alg:johnson:marking}
				\If{$\bigcap_{v \in E} \text{label}[v] = \emptyset$}{
					append vnr to $\text{label}[v]$ for $v \in E$\;
					vnr += 1\;
				}
			}
			check that new labels were added to label in this iteration\;
		}
		%	\ForEach{vertex $v$ with $\text{len(mark)} = 1$}{add vrn to $\text{mark}[v]$}
		
		\tcp{verify that $F$ exhibits Johnson symmetry}
		%$n \coloneqq \text{vnr}-1$\;
		verify that $\text{label}$ induces a bijection between $\sigma$ and $\binom{[n]}{2}$\;
		\ForEach{$i \in [n-1]$}{\label{alg:johnson:checkpermutation} 
			let $\beta$ denote the permutation of $\sigma$ induced by the Johnson action induced by $(i,i+1) \in \Sym(n)$ using label\;
			check that $\text{expand}_F(\beta)$ is a symmetry of $F$\; }
%			check that the Johnson action $j_\lambda$ induced by $\lambda \coloneqq (i,i+1) \in \Sym(n)$ on $\sigma$ using label is a symmetry of $F$\;}
		\Return{$\text{label}$}
	}
\end{algorithm}
\DecMargin{0.5em}

\begin{figure}
	\centering
		\begin{tikzpicture}[yscale=1,scale=0.6]
			\begin{scope}
			\foreach \a in {1,2,...,5}{
				\draw (\a*360/5: 1cm) node[place,minimum size=1.25mm,fill=black!50,draw=black!50] (g\a) {};
			}

			\foreach \x in {1,...,5}{
				\foreach \y in {\x,...,5}{
					\ifthenelse{\x = \y}{}{
					\draw [thick,draw=black] (g\x) to (g\y);
					}
				}
			}
			\end{scope}
			\draw [very thick, -Latex](2,0) -- (3,0);
			\begin{scope}[xshift=5cm]
				\foreach \a in {1,2,...,5}{
					\draw (\a*360/5: 1cm) node[place,minimum size=1.25mm,fill=black!50,draw=black!50] (g\a) {};
				}
	
				\foreach \x in {1,...,5}{
					\foreach \y in {\x,...,5}{
						\ifthenelse{\x = \y}{}{
						\ifthenelse{\x = 2 \AND \y = 3}{
						\draw [very thick,draw=cyan] (g\x) to (g\y);
						}{
						\draw [thick,draw=black] (g\x) to (g\y);
						}
						}
					}
				}
			\end{scope}
			\draw [very thick, -Latex](7,0) -- (8,0);
			\begin{scope}[xshift=10cm]
				\foreach \a in {1,2,...,5}{
					\draw (\a*360/5: 1cm) node[place,minimum size=1.25mm,fill=black!50,draw=black!50] (g\a) {};
				}
	
				\foreach \x in {1,...,5}{
					\foreach \y in {\x,...,5}{
						\ifthenelse{\x = \y}{}{
						\ifthenelse{\x = 2 \AND \y = 3}{
						\draw [very thick,draw=cyan] (g\x) to (g\y);
						}{
						\ifthenelse{\x = 2 \OR \y = 3 \OR \x = 3 \OR \y = 2}{
						\draw [thick,draw=orange] (g\x) to (g\y);
						}{
						\draw [thick,draw=black] (g\x) to (g\y);
						}
						}
						}
					}
				}
			\end{scope}
		\end{tikzpicture}
		\caption{Individualizing a variable $v$ in a  Johnson symmetry (represented by edges in the illustration), splits the set of edges into edges incident to $v$, and edges not incident to $v$.}
		\label{fig:johnson:ind}
\end{figure}
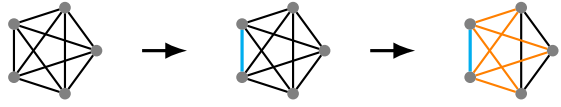

\textit{(Description of Algorithm~\ref{alg:johnson}.)} 
Suppose that $F$ exhibits a Johnson symmetry on $\sigma$. As described above, there is an (unknown) bijection $b \colon \sigma \to \binom{[n]}{2}$ (see Figure~\ref{fig:johnsondetail}). We maintain a list $\text{label}[v]$ for every $v \in \sigma$, to which we add $i \in [n]$ when we deduce that~$i \in b(v)$. If the algorithm returns a list $\text{label}$, a possible bijection $b$ is given by $b(v) = \text{label}[v]$ for all $v \in \sigma$.
Note that $b$ is only determined up to permutation of the labels, so our algorithm merely determines vertices obtaining the same label and assigns the labels consecutively.

The algorithm proceeds as follows: we apply $\IR$ to $v \in \sigma$, yielding a coloring $\pi_v$. Write $b(v) = \{i,j\}$ for some $i,j \in [n]$. The coloring $\pi_v$ has three fragments: $\{v\}$, the fragment $\sigma_v$ containing all $u \in \sigma$ with $|b(u) \cap \{i,j\}| =1$, and the remaining elements (see Figure~\ref{fig:johnson:ind}). We call the vertices in $\sigma_v$ \emph{adjacent} to $v$ and collect them in $\text{ad}[v]$. Now choose $w \in \text{ad}[v]$. We can assume $b(w) = \{j,k\}$ for some $k \notin \{i,j\}$. As before, we find the vertices adjacent to $w$ by applying $\IR$ to $w$. Individualizing both $v$ and $w$, the resulting coloring $\pi_{v,w}$ contains exactly one further singleton consisting of $y \in \sigma$ with $b(y) = \{i,k\}$. Now $\text{ad}[v] \cap \text{ad}[w] = \{y\} \cup \{u \in \sigma \colon b(u) = \{j,r\} \text{ for some } r \notin \{i,j,k\}\}$. 
%consists of $y$ and all $u \in \sigma$ with $b(u) = \{j,r\}$ for $r \in [n] \setminus \{i,j,k\}$. 
The vertices in $\text{ad}[v] \cap \text{ad}[w] \setminus \{y\}$ thus obtain the label $j$. Similarly, we determine the vertices obtaining the label $i$ or $k$. After ensuring that the labels have not been considered previously, we add them to the list $\text{label}$ for the respective vertices.

%The algorithm applies individualization-refinement to a fixed vertex $v \in \sigma$. It is checked whether this causes $\sigma$ to split into three orbits, one of them is the singleton $\{v\}$. The elements in the smaller non-singleton orbit are the purported elements incident to $v$ and they are stored in $\text{ad}[v]$. Half of these edges is incident to a label $i$ and the other half to the label $j$ of $v$. Fix one of these elements $w$ and repeat the above procedure by individualizing $w$. Now the elements in $\text{ad}[v] \cap \text{ad}[w]$ are all incident to a label $j$, and we store them in $E_j$. We now consider the coloring obtained by individualizing both $v$ and $w$. In this coloring, the elements of one fragment of size $n-3$ correspond to edges with label $i$ and the other with label $k$. At the end of the algorithm, we mark the vertices in $E_i$, $E_j$, $E_k$ by the labels that we found. We continue until every vertex has two labels.

%\markus{to a fixed vertex?}
%\markus{$\text{ad}[v] \cap \text{ad}[w]$ also contains $(i,k)$ I think... maybe description first finds these three "triangle" singletons $v,w,y$ and then proceeds from there}
%
%\sofia{make this description nicer} 

%\textit{(Runtime of Algorithm~\ref{alg:johnson}.)} \isofia{todo - the while loop is executed at most $n$ times, where $|\sigma| = n(n-1)/2$}

\textit{(Correctness of Algorithm~\ref{alg:johnson}.)} 
We again make some observations about stabilizers in Johnson groups:
\begin{lemma}\label{lemma:johnsonaux}
Let $n \in \mathbb{N}$ and consider the Johnson group $\Gamma \coloneqq \mathcal{J}_n$, acting on $2$-subsets of $[n]$. 
\begin{enumerate}
\item For $\{i,j\} \in \binom{[n]}{2}$, the orbit of $S \in \binom{[n]}{2}$ under the stabilizer $\Gamma_{\{i,j\}}$ of $\{i,j\} \in \binom{[n]}{2}$ is given by 
\[S^{\Gamma_{\{i,j\}}} = 
\begin{cases}
\{S\} &\text{if } S = \{i,j\} \\
\{T \in \binom{[n]}{2} \colon |T \cap \{i,j\}| = 1\} & \text{if } |S \cap \{i,j\}| = 1 \\
\{T \in \binom{[n]}{2} \colon T \cap \{i,j\} = \emptyset\} &\text{if } S \cap \{i,j\} = \emptyset.	
\end{cases}	
\]
\item For $k \neq i,j$, 
%we have $\Gamma_{\{i,j\}} \cap \Gamma_{\{i,k\}} = \{\alpha \in \Sym(n) \colon i^\alpha = i,\, j^\alpha = j,\, k^\alpha = k\}$. The 
the orbit of $S \in \binom{[n]}{2}$ under $\Gamma_{\{i,j\}} \cap \Gamma_{\{i,k\}}$ is given by
\[
S^{\Gamma_{\{i,j\}} \cap \Gamma_{\{i,k\}}} = \begin{cases}
\{S\} &\text{if } S \in \{\{i,j\}, \{i,k\}, \{j,k\}\} \\
\{\{i,r\} \colon r \in [n] \setminus \{i,j,k\}\} & \text{if } S = \{i,s\} \text{ for some } s\in [n] \setminus \{i,j,k\} \\	
\{\{j,r\} \colon r \in [n] \setminus \{i,j,k\}\} & \text{if } S = \{j,s\} \text{ for some } s\in [n] \setminus \{i,j,k\} \\
\{\{k,r\} \colon r \in [n] \setminus \{i,j,k\}\} & \text{if } S = \{k,s\} \text{ for some } s\in [n] \setminus \{i,j,k\} \\
\{S \in \binom{[n]}{2} \colon S \cap \{i,j,k\} = \emptyset\} &\text{if } S \cap \{i,j,k\} = \emptyset.
\end{cases}	
\]
\end{enumerate}
\end{lemma}	

\begin{proof}
		$\null$
\begin{enumerate}
	\item Let $S \in \binom{[n]}{2}$. If $S = \{i,j\}$, the orbit $S^{\Gamma_{\{i,j\}}}$ consists only of $S$ by definition of the stabilizer. Now suppose that $|S \cap \{i,j\}| = 1$ holds. Without loss of generality, let $S = \{i,r\}$ for some $r \in [n] \setminus \{i,j\}$. Let $\pi \in \Gamma_{\{i,j\}}$. Either $\pi$ fixes $i$ and $j$, in which case we have $S^\pi = \{i, r'\}$ for some $r' \in [n] \setminus \{i,j\}$, or $\pi$ interchanges $i$ and $j$, in which case we have $S^\pi = \{j,r'\}$ for some $r' \in [n] \setminus \{i,j\}$. In both cases, we have $|S^\pi \cap \{i,j\}| = 1$. On the other hand, it is easy to see that for every set $T \in \binom{[n]}{2}$ with $|T \cap \{i,j\}| = 1$, there exists $\pi \in \Gamma_{\{i,j\}}$ with $S^\pi = T$. The description of $S^{\Gamma_{\{i,j\}}}$ in the case $S \cap \{i,j\} = \emptyset$ can be derived analogously. 

\item Note that an element in $\Gamma_{\{i,j\}} \cap \Gamma_{\{i,k\}}$ fixes or interchanges the labels $i$ and $j$, and at the same time fixes or interchanges the labels $i$ and $k$. This is only possible if it fixes all of $i$, $j$ and $k$. The structure of the orbits then follows similarly to the first claim. \qedhere
\end{enumerate}
\end{proof}

We now prove that the algorithm always returns correct symmetries of $F$ and that in case the model graph is Tinhofer, the algorithm is guaranteed to detect that $F$ exhibits a Johnson symmetry on the input set $\sigma$.
\begin{restatable}{theorem}{corrjohnson}
Let $F$ be a SAT formula. 
\begin{enumerate}
	\item If Algorithm~\ref{alg:johnson} returns a list \emph{label} of labels in $[n]$, then for every element in $\mathcal{J}_n$, the induced permutation of $\sigma$ according to label, expanded to $\neg \sigma$, is a symmetry of $F$.
	\item If $F$ exhibits a Johnson symmetry with Johnson group $\mathcal{J}_n$ with $n \geq 8$ on $\sigma$ and $G(F)$ is a Tinhofer graph, then Algorithm~\ref{alg:johnson} detects this structure and returns a corresponding labeling of the literals in $\sigma$ by 2-subsets of $[n]$. 
\end{enumerate}	
%	Given a Tinhofer graph $G$ and a candidate orbit $\sigma$, Algorithm~\ref{alg:johnson} correctly decides whether $G$ exhibits a Johnson symmetry on $\sigma$ with $n \geq 8$, where $|\sigma| = \binom{n}{2}$, and if so, returns a mapping of the vertices of $\sigma$ to 2-subsets of a domain $[n]$. 
\end{restatable}
\begin{proof}
	The last part of Algorithm~\ref{alg:johnson} ensures that the Johnson action $j_\pi$ induced by a transposition $\pi \coloneqq (i,i+1) \in \mathcal{J}_n$ by permuting the elements in $\sigma$ according to their labels in \emph{label} is a symmetry of $F$ when expanded to $\neg \sigma$. By suitably composing these transpositions, it follows that every element of $\mathcal{J}_n$ induces a symmetry of $F$ in this way. 
	
	Now suppose that $F$ exhibits a Johnson symmetry with Johnson group $\mathcal{J}_n$ with $n \geq 8$ (i.e., $|\sigma| \geq 28$). Furthermore, assume that $G(F)$ is Tinhofer. In particular, there is a bijection $b \colon \sigma \to \binom{[n]}{2}$ (see Figure~\ref{fig:johnsondetail}). We claim that when the algorithm terminates, there is a permutation $\tau \in \Sym(n)$ of the label set $[n]$ such that we have $\text{label}[v] = \{\tau(i),\tau(j)\}$ if $b(v)= \{i,j\}$. Note that the bijection $b$ itself is determined only up to permutation of the labels. Again, for the sake of clarity, we refer to the elements of~$[n]$ as \emph{labels} and reserve the term \emph{vertices} for the vertices of the graph $G(F)$.  
	
	The individualization of a vertex $v$ with $b(v) = \{i,j\}$ (Line \ref{alg:johnson:firstindividualization}) leads to a color partition with three fragments since $G(F)$ is Tinhofer (see Lemma~\ref{lemma:johnsonaux} and Figure~\ref{fig:johnson:ind}). 
	The smaller non-singleton fragment is $\sigma_v = \{u \in \sigma \colon |b(u) \cap \{i,j\}| = 1\}$. For this, note that $|\sigma_v| = 2(n-2)$ holds and that we have $n \geq 8$ by assumption. 
	The list $\text{ad}[v]$ (Line~\ref{alg:johnson:adv}) then consists of all vertices $u \in \sigma$ with $b(u) = \{i,r\}$ or $b(u) = \{j,r\}$ with $r \in [n] \setminus \{i,j\}$.
	
	Now let $w \in \text{ad}[v]$. Up to this point, the labels $i$ and $j$ are interchangeable, so we may assume $b(w) = \{j,k\}$ for some $k \in [n] \setminus \{i,j\}$. We repeat the above procedure with $w$ in place of $v$. In particular, $\text{ad}[w]$ (Line~\ref{alg:johnson:adw}) contains all vertices $u \in \sigma$ with $b(u) = \{j,r\}$ or $b(u) = \{k,r\}$ for $r \in [n] \setminus \{j,k\}$. 
	
	Finally we individualize both $v$ and $w$ to obtain the coloring $\pi_{v,w}$. The fragments are given by Lemma~\ref{lemma:johnsonaux}. In particular, we obtain $b(y) = \{i,k\}$. Apart from $y$, the intersection $\text{ad}[v] \cap \text{ad}[w]$ contains all vertices $u \in \sigma$ with $b(u) = \{j,r\}$ for $r \in [n] \setminus \{i,j,k\}$, and we add them to $E_j$ (Line~\ref{alg:johnson:ej}). Similarly, we construct the sets $E_i$ and $E_k$ (Lines~\ref{alg:johnson:ei} and~\ref{alg:johnson:ek}).
	%. By Lemma~\ref{lemma:johnsonaux}, $\pi_{v,w}$ contains three fragments of size $n-3$. Note that $\sigma_k$ consists of the elements $\{k,r\}$ with $r \notin \{i,j,k\}$ and that $\sigma_i$ consists of the elements $\{i,r\}$ with $r \notin \{i,j,k\}$. 
	
	From this explicit description, it is clear that $u \in \sigma$ is added to $E_i$ precisely if $i \in b(u)$ (similarly for $E_j$ and $E_k$). In particular, for distinct vertices $u_1, u_2 \in E_i$, we have $b(u) \cap b(v) = \{i\}$. Thus if the lists $\text{label}[u]$ for $u \in E_i$ have a common entry, the label $i$ has been considered before (recall that $|E_i| > 1$ holds). Otherwise, we add the current vertex number \emph{vnr} to $\text{label}[u]$ for all $u \in E_i$ (Line~\ref{alg:johnson:marking}) and set $\tau(i) = \text{vnr}$. This way, $\text{label}[u]$ remains duplicate-free and only ever contains labels $\tau(l)$ for $l  \in b(u)$. In particular, we always maintain the property $|\text{label}[u]| \leq 2$. In each iteration of the while loop, one of the labels $i$ and $j$ was not considered before (due to $|\text{label}[v]|\leq 1$). In particular, the loop is executed at most $n$ times. When it stops, we have $|\text{label}[v]| = 2$ for all vertices $v$.
	%
	%By the first claim, if the algorithm returns a list \emph{mark} of labels, then for every $\pi \in \Sym(n)$, the corresponding Johnson action $j_\pi$ on $\sigma$ is indeed a symmetry of $F$. Let $J = \{j_\pi \colon \pi \in \Sym(n)\}$ be the group of these symmetries. Let $\alpha \in \Aut(F)$. We claim that the action of~$\alpha$ on $\sigma$ coincides with that of an element in $J$. 
	%By multiplying $\alpha$ with a suitable element of $J$, we may assume that $\alpha$ fixes a vertex $v \in \sigma$. Since the point stabilizer $\Aut(F)_v$ has three orbits on $\sigma$, $\alpha$ permutes the set $\text{ad}[v]$ of vertices adjacent to $v$. 
	%\markus{pointwise stabilizers notation}
	%By multiplying with a suitable element of $J$, we may assume that $\alpha$ fixes $\text{ad}[v]$ pointwise. 
	%By assumption, applying $\IR$ to the vertices in $\text{ad}[v] \cup \{v\}$ leads to a color partition, where every literal of the disjoint direct factor is a singleton. 
	%Hence $\alpha$ acts as the identity on $\sigma$. This shows that $\Aut(F)|_\sigma = J|_\sigma$ holds, that is, $F$ exhibits a Johnson symmetry $\mathcal{J}_n$ on~$\sigma$.
	\end{proof}
	
\textbf{Johnson Action on Row Symmetry.}
Quite commonly, SAT instances which search for a graph, will search for a graph with a certain \emph{vertex property}.
For example, when asking for a $k$-colorable graph, there will be (interchangeable) colors attached to each vertex of the graph.
In order to detect a corresponding symmetry structure, we want to detect blocks which correspond to the labels in the Johnson domain.
The detection works by stabilizing vertices in other orbits, and checking whether they split apart the Johnson orbit precisely into the vertices marked with a particular label, and a remainder.  
If so, these blocks are collected and considered in our overall Johnson action.
Finally, we run row symmetry detection on the collected blocks.

\section{Implementation} \label{sec:implementation}
We now give an overview of our new symmetry breaking tool \textsc{satsuma}.
The input of our algorithm is a CNF formula $F$. 
The output is a symmetry breaking constraint for $F$.
We first discuss the breaking constraints produced for a given detected structure. 

\textbf{Breaking Constraints.} 
We produce lex-leader constraints for each detected structure: we use precisely the automorphisms constructed in Algorithm~\ref{alg:row}, Algorithm~\ref{alg:rowcolumn}, and Algorithm~\ref{alg:johnson}.
Before we can produce lex-leader constraints, we must however fix an ordering on the variables.
The ordering used for matrix models simply orders the matrix row-by-row.
For Johnson groups, we begin with the vertices of the first label (see Algorithm~\ref{alg:johnson}), then the remaining vertices of the second label, and so forth.

\textbf{High-level Algorithm.} The high-level algorithm proceeds as follows:\\
\textit{(Step 1.)} Construct a model graph from the given CNF formula.\\
\textit{(Step 2.)} Run the algorithms described in the previous section in the following order: Johnson groups (Algorithm~\ref{alg:johnson}), row-column symmetry (Algorithm~\ref{alg:rowcolumn}), row interchangeability (Algorithm~\ref{alg:row}). Whenever a structure is found, all orbits covered by the structure are marked. The subsequent analysis only considers \emph{unmarked} orbits. For each structure, symmetry breaking constraints are constructed as described above.
Lastly, we maintain a vertex coloring of the model graph, which we call the \emph{remainder coloring}: this coloring restricts the symmetries of the model graph to symmetries not yet covered by detected structures.\\
\textit{(Step 3.)} Run symmetry detection on the graph colored with the remainder coloring. Then, the \emph{binary clause} heuristic of \textsc{BreakID} is applied for all variables not yet ordered by already produced lex-leader constraints: a stabilizer chain of the automorphism group is approximated, and for each stabilized variable $x$ a short lex-leader constraint for each other literal $y$ of its orbit is produced, i.e., essentially the binary constraint $x \leq y$ (see \cite{DBLP:conf/sat/Devriendt0BD16} for a detailed description).
	Lastly, a lex-leader constraint for each generator is produced. 

\textbf{Implementation.} 
%Our implementation of the above algorithms is called \textsc{satsuma}. 
The tool is written in \textsf{C++}, and is freely available as open source software \cite{satsuma:github}. %´\sofia{in my logic, this sentence should come earlier, maybe at the beginning of Chapter 4}
The tool \textsc{dejavu} \cite{dejavu:webpage, DBLP:conf/esa/AndersS21, DBLP:conf/wea/AndersSS23} is used for providing general-purpose symmetry detection, the individualization-refinement framework, and data structures for symmetries.
Significant parts of the implementation, in particular the generation of lex-leader constraints and binary clauses, are reverse-engineered from \textsc{BreakID}.
Our reimplementation of these routines differs in two crucial aspects from the original one:
first, \textsc{BreakID} uses the symmetry detection tool \textsc{saucy} \cite{DBLP:conf/dac/DargaLSM04} instead of \textsc{dejavu}.
Second, we use different data structures and algorithms for the handling of symmetry.

\section{Benchmarks} \label{sec:benchmarks}
We compare the state-of-the-art static symmetry breaking tool \textsc{BreakID} (version 2.6) to \textsc{satsuma}. 
\begin{figure}
	\centering
	\small
	\begin{tabular}{|ll|cc|ccc|ccc|}\hline
        \multicolumn{2}{|c|}{family}& \multicolumn{2}{c|}{\textsc{CMS}} & \multicolumn{3}{c|}{\textsc{BreakID}+\textsc{CMS}} & \multicolumn{3}{c|}{\textsc{satsuma}+\textsc{CMS}}\\ 
		name & size & solved & avg & prep & solved & avg & prep & solved & avg \\ \hline
		\textsf{channel} &10&2&484.99&4.727&\textbf{10}&\textbf{0.032}&\textbf{0.404}&\textbf{10}&0.033
		\\
		\textsf{cliquecolor} & 20 & 2&574.734&0.129&13&228.998&\textbf{0.058}&\textbf{20}&\textbf{0.845}
		\\
		\textsf{coloring} & 55 & 21&377.338&42.12*&26&317.32&\textbf{1.071}&\textbf{27}&\textbf{307.632}\\
		\textsf{fpga} & 10 & 6&321.596&0.035&\textbf{10}&0.01&\textbf{0.01}&\textbf{10}&\textbf{0.008}\\
		\textsf{md5} & 11 & 5&358.616&0.635&5&359.382&\textbf{0.548}&\textbf{6}&\textbf{349.171} \\
		\textsf{php} & 10 & 3&423.266&6.337&\textbf{10}&0.043&\textbf{0.128}&\textbf{10}&\textbf{0.036}\\
		\textsf{ramsey} & 7 & 2&428.613&1.681&3&343.086&\textbf{0.394}&\textbf{5}&\textbf{235.27} \\
		\textsf{urquhart} & 6 & \textbf{6}&0.768&0.14&\textbf{6}&\textbf{0.008}&\textbf{0.032}&\textbf{6}&0.066\\\hline
	\end{tabular}
	\caption{Benchmarks comparing \textsc{BreakID} to \textsc{satsuma}, using the solver \textsc{CryptoMiniSAT} (\textsc{CMS}). The timeout is $600$ seconds, all times are given in seconds. The columns ``prep'' denote the average time used for symmetry breaking. Columns ``solved'' refer to the number of solved instances by \textsc{CMS}, and ``avg'' is the average time spent by \textsc{CMS}. *\textsc{BreakID} could not compute the symmetry breaking constraints of two \textsf{coloring} instances within the timeout. We declared these as a timeout for the SAT benchmarks (but the other configurations also timed out on these instances).} \label{fig:benchmarks:main}
\end{figure}
\begin{figure}
	\centering
	\begin{tabular}{|ll|cc|ccc|ccc|}\hline
        \multicolumn{2}{|c|}{family}& \multicolumn{2}{c|}{\textsc{CaD}} & \multicolumn{3}{c|}{\textsc{BreakID}+\textsc{CaD}} & \multicolumn{3}{c|}{\textsc{satsuma}+\textsc{CaD}}\\ 
		name & size & solved & avg & prep & solved & avg & prep & solved & avg \\ \hline
		\textsf{channel} &10&2&494.226&4.727&\textbf{10}&\textbf{0.077}&\textbf{0.404}&\textbf{10}&\textbf{0.077}
		\\
		\textsf{cliquecolor} & 20 & 9&442.373&0.129&13&216.999&\textbf{0.058}&\textbf{20}&\textbf{0.2}
		\\
		\textsf{coloring} & 55 & 20&393.864&42.12*&26&316.779&\textbf{1.071}&\textbf{28}&\textbf{301.783}
		\\
		\textsf{fpga} & 10 & 5&391.29&0.035&\textbf{10}&\textbf{0.008}&\textbf{0.01}&\textbf{10}&0.025\\
		\textsf{md5} & 11 & \textbf{6}&339.378&0.635&\textbf{6}&343.324&\textbf{0.548}&\textbf{6}&\textbf{324.716}
		\\
		\textsf{php} & 10 & 3&422.976&6.337&\textbf{10}&\textbf{0.085}&\textbf{0.128}&\textbf{10}&0.1\\
		\textsf{ramsey} & 7 & 2&428.583&1.681&3&342.908&\textbf{0.394}&\textbf{5}&\textbf{192.299} \\
		\textsf{urquhart} & 6 & 2&449.622&0.14&\textbf{6}&\textbf{0.005}&\textbf{0.032}&\textbf{6}&0.052\\\hline
	\end{tabular}
	\caption{Benchmarks comparing \textsc{BreakID} to \textsc{satsuma}. The SAT solver used is \textsc{CaDiCaL} (\textsc{CaD}). The timeout used is $600$ seconds. The columns ``prep'' refer to the time in seconds used to compute the symmetry breaking constraint. Columns ``solved'' refer to the number of solved instances by \textsc{CaD}, and ``avg'' is the average time used by \textsc{CaD} (\emph{excluding} the time used for symmetry breaking). *\textsc{BreakID} could not compute the symmetry breaking constraints of two \textsf{coloring} instances within the timeout.} \label{fig:benchmarks:cadical}
\end{figure}

As SAT solvers, we use \textsc{CryptoMiniSAT} \cite{DBLP:conf/sat/SoosNC09} and \textsc{CaDiCaL} \cite{BiereFazekasFleuryHeisinger-SAT-Competition-2020-solvers}.
%The benchmarks using \textsc{CaDiCaL} are in Appendix~\ref{sec:cadical}. 
The benchmarks using \textsc{CaDiCaL} largely concur with the \textsc{CryptoMiniSAT} benchmarks, and our descriptions will focus on the results using \textsc{CryptoMiniSAT}.
The timeout for all benchmarks is $600$ seconds. 
We separately measure the time spent on symmetry breaking itself, and SAT solving. 
All benchmarks ran sequentially on an Intel Core i7 9700K with 64GB of RAM on Ubuntu 20.04.

\textbf{Benchmark Instances.}
We run benchmarks on a variety of well-established instance families exhibiting symmetry (see Figure~\ref{fig:benchmarks:main}).
The sets \textsf{coloring}, \textsf{urquhart}, \textsf{fpga}, \textsf{md5}, and \textsf{channel} are part of the distribution of \textsc{BreakID} \cite{DBLP:conf/sat/Devriendt0BD16}.
We generate pigeonhole principle (\textsf{php}) instances, Ramsey instances, and clique coloring instances using the tool \textsc{cnfgen} \cite{DBLP:conf/sat/LauriaENV17}.
The set of parameters for clique coloring is similar to \cite{DBLP:journals/jsc/JunttilaKKK20}, but we added larger instances.
All instances are unsatisfiable.
Individual instances and results are listed in Appendix~\ref{sec:benchmarkdetail}.

Regarding the detected symmetry structures of these instances, we detect Johnson symmetry on the \textsf{ramsey} and \textsf{cliquecolor} families.
On \textsf{php}, \textsf{channel}, and \textsf{fpga}, \textsc{satsuma} detects row-column symmetry, and \textsc{BreakID} corresponding row interchangeability (see also \cite{DBLP:journals/constraints/Sabharwal09,DBLP:conf/sat/Devriendt0BD16}).
The \textsf{coloring} instances exhibit a variety of different symmetries, but in particular also row symmetry \cite{DBLP:conf/sat/Devriendt0BD16}.
In \textsf{urquhart} and \textsf{md5}, no structure is detected by either of the tools.  

Regarding our choice of benchmark instances, we stress that our main goal is to observe whether detecting richer structures can improve performance compared to existing approaches.

\textbf{SAT Benchmarks.}
An overview of the results can be found in Figure~\ref{fig:benchmarks:main} (for \textsc{CaDiCaL}, see Figure~\ref{fig:benchmarks:cadical}).
Considering the results, we observe that \textsc{satsuma}  solves more instances, and solving times are considerably lower on average on the \textsf{cliquecolor} and \textsf{ramsey} instances. 
We recall that these instance families exhibit Johnson symmetry.
On all sets with row and row-column symmetry, that is \textsf{channel}, \textsf{coloring}, \textsf{fpga}, and \textsf{php}, we observe that solved instances and average solving times are comparable.
On \textsf{coloring}, we observe that \textsc{satsuma} solves one more instance than \textsc{BreakID} (and two more using \textsc{CaDiCaL}).
For \textsf{urquhart}, both \textsc{satsuma} and \textsc{BreakID} rely on the binary clause strategy. 
The results indicate that \textsc{BreakID} is more effective in breaking symmetry, which is however outweighed by the faster runtime of \textsc{satsuma}.
The \textsf{md5} instances only contain a single non-trivial symmetry.
Here, \textsc{satsuma} produces more breaking clauses, and we observe a consistent albeit marginal speedup.
It should be mentioned that it does however seem plausible that the observed speed-up may be due to shuffling of literals in clauses, or other factors.

In particular, we point out that \textsc{satsuma} compares favorably on instance families which exhibit Johnson symmetry.
We believe this to be due to our detection of Johnson symmetry and the subsequent generation of more favorable constraints. 
Crucially, on all successfully solved instances of \textsf{cliquecolor} and \textsf{ramsey}, the \emph{remainder contains no symmetry}: all symmetries are detected and in turn broken solely using the algorithms of this paper, and no general-purpose symmetry detection and breaking is applied.

We observe that the average time spent computing the symmetry breaking constraints is lower on all families for \textsc{satsuma}. 
A more in-depth analysis follows below.

\begin{figure}
	\centering
	\scalebox{0.5}{
		% This file was created with tikzplotlib v0.9.17.
\begin{tikzpicture}

\definecolor{color0}{rgb}{0.12156862745098,0.466666666666667,0.705882352941177}
\definecolor{color1}{rgb}{1,0.498039215686275,0.0549019607843137}

\begin{axis}[
legend cell align={left},
legend style={
  fill opacity=0.8,
  draw opacity=1,
  text opacity=1,
  at={(0.03,0.97)},
  anchor=north west,
  draw=white!80!black
},
log basis y={10},
tick align=outside,
tick pos=left,
title={{overhead\_php}},
x grid style={white!69.0196078431373!black},
xlabel={pigeons},
xmin=3, xmax=157,
xtick style={color=black},
y grid style={white!69.0196078431373!black},
ylabel={computation time (s)},
ymin=0.00274792065127367, ymax=757.144430379844,
ymode=log,
ytick style={color=black}
]
\addplot [semithick, color0, opacity=0.5, mark=*, mark size=1.5, mark options={solid}]
table {%
10 0.004856069
20 0.022671271
30 0.042603877
40 0.045966922
50 0.083128178
60 0.143757423
70 0.257087781
80 0.354929479
90 0.618676624
100 0.912411576
110 1.288499099
120 1.814128264
130 2.400829717
140 3.015157426
150 4.316872369
};
\addlegendentry{\textsc{satsuma}}
\addplot [semithick, color1, opacity=0.5, mark=+, mark size=1.5, mark options={solid}]
table {%
10 0.007406826
20 0.051406048
30 0.223962074
40 0.683375401
50 1.777963855
60 3.70048651
70 8.105637691
80 15.352032598
90 30.734331358
100 49.828880337
110 87.392463599
120 139.238977594
130 200.514627469
140 298.248344213
150 428.447951674
};
\addlegendentry{\textsc{BreakID}}
\end{axis}

\end{tikzpicture}
	}
	\scalebox{0.5}{
		% This file was created with tikzplotlib v0.9.17.
\begin{tikzpicture}

\definecolor{color0}{rgb}{0.12156862745098,0.466666666666667,0.705882352941177}
\definecolor{color1}{rgb}{1,0.498039215686275,0.0549019607843137}

\begin{axis}[
legend cell align={left},
legend style={
  fill opacity=0.8,
  draw opacity=1,
  text opacity=1,
  at={(0.03,0.97)},
  anchor=north west,
  draw=white!80!black
},
log basis y={10},
tick align=outside,
tick pos=left,
title={{overhead\_clqcolor}},
x grid style={white!69.0196078431373!black},
xlabel={vertices},
xmin=-4.5, xmax=314.5,
xtick style={color=black},
y grid style={white!69.0196078431373!black},
ylabel={computation time (s)},
ymin=0.00296426880265207, ymax=378.050775006307,
ymode=log,
ytick style={color=black}
]
\addplot [semithick, color0, opacity=0.5, mark=*, mark size=1.5, mark options={solid}]
table {%
10 0.005081214
20 0.010292257
30 0.021225147
40 0.041169493
50 0.066609579
60 0.102679797
70 0.151359152
80 0.212453379
90 0.288810498
100 0.374300582
110 0.486403331
120 0.633869386
130 0.78520752
140 0.987682348
150 1.203782925
160 1.447614248
170 1.744458756
180 2.075359983
190 2.45178132
200 2.871608497
210 3.300728302
220 3.829861562
230 4.382301722
240 5.091619849
250 5.84234697
260 6.775428854
270 7.698482535
280 9.648613899
290 10.490529329
300 11.666321415
};
\addlegendentry{\textsc{satsuma}}
\addplot [semithick, color1, opacity=0.5, mark=+, mark size=1.5, mark options={solid}]
table {%
10 0.005058151
20 0.017078724
30 0.047900726
40 0.112127237
50 0.221330277
60 0.386309518
70 0.651473076
80 1.034297458
90 1.638270294
100 2.591032984
110 3.271883564
120 4.650230108
130 6.53855365
140 8.860807833
150 11.278455594
160 15.619426945
170 18.883515665
180 23.648150765
190 31.651993442
200 36.264616285
210 48.537539875
220 63.069699111
230 72.289665719
240 85.394221814
250 105.368703415
260 119.685769765
270 140.44407193
280 156.178097723
290 182.124929841
300 221.552128074
};
\addlegendentry{\textsc{BreakID}}
\end{axis}

\end{tikzpicture}
	}
	\scalebox{0.5}{
		% This file was created with tikzplotlib v0.9.17.
\begin{tikzpicture}

\definecolor{color0}{rgb}{0.12156862745098,0.466666666666667,0.705882352941177}
\definecolor{color1}{rgb}{1,0.498039215686275,0.0549019607843137}

\begin{axis}[
legend cell align={left},
legend style={
  fill opacity=0.8,
  draw opacity=1,
  text opacity=1,
  at={(0.97,0.03)},
  anchor=south east,
  draw=white!80!black
},
log basis y={10},
tick align=outside,
tick pos=left,
title={{overhead\_urquhart}},
x grid style={white!69.0196078431373!black},
xlabel={vertices},
xmin=-7, xmax=367,
xtick style={color=black},
y grid style={white!69.0196078431373!black},
ylabel={computation time (s)},
ymin=0.0018988610345399, ymax=1096.57850791832,
ymode=log,
ytick style={color=black}
]
\draw[line width=4pt, red!30] (axis cs:-20,600) -- (axis cs:600,600);
\addplot [semithick, color0, opacity=0.5, mark=*, mark size=1.5, mark options={solid}]
table {%
10 0.004472619
20 0.006225698
30 0.010097407
40 0.011576133
50 0.014792831
60 0.020483655
70 0.027513418
80 0.030854421
90 0.040749406
100 0.044547332
110 0.053750191
120 0.06854882
130 0.075363302
140 0.083704754
150 0.094183895
160 0.119961569
170 0.117929302
180 0.135461816
190 0.150290415
200 0.184899166
210 0.208687465
220 0.228152047
230 0.250908479
240 0.229632558
250 0.246717593
260 0.26752934
270 0.231907607
280 0.286333986
290 0.364589536
300 0.179264148
310 0.234187186
320 0.41442788
330 0.233537071
340 0.235195824
350 0.590134491
};
\addlegendentry{\textsc{satsuma}}
\addplot [semithick, color1, opacity=0.5, mark=+, mark size=1.5, mark options={solid}]
table {%
10 0.003470417
20 0.006855121
30 0.020279546
40 0.019763563
50 0.03368733
60 0.058799986
70 0.075584059
80 0.106549415
90 0.147159072
100 0.191472492
110 0.264087796
120 0.314460819
130 0.409854694
140 0.508914081
150 0.588603825
160 9.126428594
170 0.981623409
180 1.041428702
190 536.325936251
200 53.014165099
210 523.51514111
220 600
230 215.502935723
240 585.961834761
250 2.684199217
260 3.099926938
270 3.388681156
280 31.956253392
290 5.218896063
300 4.627351169
310 600
320 66.942211632
330 600
340 600
350 600
};
\addlegendentry{\textsc{BreakID}}
\end{axis}

\end{tikzpicture}
	}
	\caption{Benchmarks comparing the computational overhead of \textsc{BreakID} to \textsc{satsuma}. The shown computation time is the time spent computing symmetry breaking constraints for an instance using the respective tool. The red bar indicates the timeout of $600$ seconds.} \label{fig:benchmarks:overhead}
\end{figure}
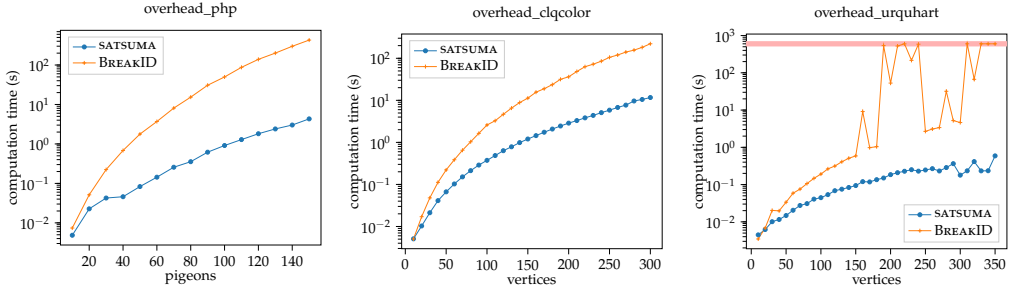

\textbf{Computational Overhead.} 
We conduct further benchmarks to gauge the computational overhead incurred by \textsc{BreakID} and \textsc{satsuma}.
We test three different benchmark families: \textsf{php}, \textsf{cliquecolor}, and \textsf{urquhart} (generated using \textsc{cnfgen}).
For \textsf{php}, we increase the number of pigeons from $10$ to $150$ (with $n-1$ holes, respectively).
For \textsf{cliquecolor}, we increase the number of vertices of the prospective graph from $10$ to $300$ (the size of the clique is $3$ and number of colors $2$).
In \textsf{urquhart}, we use random $5$-regular graphs, increasing the number of vertices from $10$ to $350$.
We chose these instance families such that they cover the different symmetry detection routines in \textsc{satsuma}:
the family \textsf{php} essentially measures the runtime of our row-column routine, \textsf{cliquecolor} that of the Johnson routine, and \textsf{urquhart} uses general purpose symmetry detection, followed by the binary clause strategy.

Figure~\ref{fig:benchmarks:overhead} summarizes the results. 
In all instance families, the data suggest that \textsc{satsuma} asymptotically scales better than \textsc{BreakID}.
These results match our observations regarding overhead from the first part of the benchmarks (see Figure~\Ref{fig:benchmarks:main}).  

We believe there are multiple reasons why \textsc{satsuma} runs faster than \textsc{BreakID}.
First, our new algorithms of Section~\ref{sec:detectionalgorithms} verify symmetries on the CNF formula instead of the model graph.
This is advantageous because symmetries of the CNF only explicitly map literals, whereas symmetries of the model graph also explicitly map clauses.
Second, most routines in our implementation run proportional in the size of the \emph{support} of symmetries, as opposed to the number of literals of $F$. 
Third, for general-purpose symmetry detection, \textsc{dejavu} seems to be more efficient in computing automorphism groups of SAT instances than \textsc{saucy} \cite{dejavu:webpage}.

We mention that in the \textsf{urquhart} instances, the outliers with high running time seem to be due to \textsc{saucy} taking a long time to compute symmetries for \textsc{BreakID}.
On the other hand, in these cases, we observe that the symmetries as returned by \textsc{dejavu} are less suitable for the binary clause heuristic, leading to fewer produced clauses.
This could however be easily alleviated by a strengthening of the heuristic (e.g., by sometimes applying the Schreier-Sims algorithm for stabilizers as already pointed out in \cite{DBLP:conf/sat/Devriendt0BD16}).  

% \subsection{Impact of Structure Detection}

% \begin{figure}
% 	\centering
% 	\begin{tabular}{|ll|ll|ll|ll|}\hline
%         \multicolumn{2}{|c|}{family}& \multicolumn{2}{c|}{\textsc{satsuma-no-struct}} & \multicolumn{2}{c|}{\textsc{satsuma-only-struct}}  & \multicolumn{2}{c|}{\textsc{satsuma}}\\ 
% 		name & size & solved & avg & solved & avg & solved & avg \\
% 		channel & 10 & & & &&& \\
% 		cliquecolor & 14 & & & &&& \\
% 		color & 42 & & & & &&\\
% 		fpga & 10 & & & && &\\
% 		md5 & 11 & & & && &\\
% 		php & 9 & & & && &\\
% 		ramsey & & & &  && &\\
% 		urquhart & 6 & & & && &\\\hline
% 	\end{tabular}
% 	\caption{Benchmarks comparing \textsc{BreakID} to \textsc{satsuma}. The SAT solver used is \textsc{CryptoMiniSAT}. The timeout for the symmetry breaking tools is $360$ seconds, and for solving $600$ seconds.}
% \end{figure}

\section{Conclusions and Future Work} \label{sec:conclusions}
We described a new structure-based approach to symmetry breaking, and demonstrated the effectiveness of our implementation \textsc{satsuma}.
%As for future work, our goal is to further develop \textsc{satsuma}.
There seem to be many promising directions in which the present work could be expanded:
\begin{itemize}
	\item Detect more group structures: in particular, a more generic approach to detect aggregates of groups would be of great interest. Another interesting case might be the symmetries of the family \textsf{urquhart}, which are isomorphic to $C_2^k$ and have been studied previously \cite{DBLP:journals/amai/LuksR04}.
	\item Consider other breaking approaches for certain group structures. So far, we used the knowledge of group structures to pick out automorphisms, for which off-the-shelf lex-leader constraints are generated.
	Since optimal handling of row-column symmetry and Johnson symmetry seems infeasible with lex-leader constraints \cite{DBLP:journals/amai/LuksR04}, other breaking constraints could lead to better results. Moreover, Johnson symmetry allows the use of symmetry reduction developed specifically for graph generation \cite{DBLP:conf/cp/CodishGIS16, DBLP:journals/constraints/CodishMPS19, DBLP:conf/cp/KirchwegerS21}. 
	\item Improved techniques for handling of the ``remainder''. As already pointed out in \cite{DBLP:conf/sat/Devriendt0BD16}, one potential direction would be to apply the random Schreier-Sims algorithm \cite{seress_2003} to produce more small symmetry breaking clauses.  
	\item An enticing feature is proof-logging, as was recently introduced to \textsc{BreakID} \cite{DBLP:journals/jair/BogaertsGMN23}.
	\item The new detection algorithms could be applied in other domains as well: for example, seeing as row interchangeability is successfully used in MIP, it seems only natural that MIP instances may also contain richer structures.
	\item Sometimes symmetries are not present in a compiled CNF of a given problem (as, e.g., analyzed in \cite{DBLP:journals/jsc/JunttilaKKK20}). A possible remedy is to allow the user to provide an auxiliary graph that models the original symmetry (see \cite{DBLP:journals/jsc/JunttilaKKK20}), and the methods proposed in this paper should generalize to this setting. 
\end{itemize}

\section*{Funding}
The research leading to these results has received funding from the European Research Council (ERC) under the European Union's Horizon 2020 research and innovation programme
(EngageS: grant agreement No. 820148). Sofia Brenner additionally received funding from the German Research Foundation DFG
(SFB-TRR 195 “Symbolic Tools in Mathematics and their Application”).

\bibliography{main}
\bibliographystyle{plain}

\newpage
\appendix
{\Huge Appendix}

\section{Benchmark Results for Individual Instances} \label{sec:benchmarkdetail}
The results for individual instances can be found in Figure~\ref{fig:ind:channel}, Figure~\ref{fig:ind:clqcolor}, Figure~\ref{fig:ind:coloring}, Figure~\ref{fig:ind:fpga}, Figure~\ref{fig:ind:md5}, Figure~\ref{fig:ind:php}, and Figure~\ref{fig:ind:ramsey}. The figures contain average solving times for both \textsc{CryptoMiniSAT} (\textsc{CMS}) and and \textsc{CaDiCaL} (\textsc{CaD}).

\begin{figure}[p]
	\centering
	\small
	\begin{tabular}{|l|c|c|ccc|ccc|}\hline
		 & \textsc{CMS} & \textsc{CaD} & \multicolumn{3}{c|}{\textsc{BreakID}} & \multicolumn{3}{c|}{\textsc{satsuma}}\\
		 instance & solve & solve & prep & \textsc{CMS}& \textsc{CaD} & prep & \textsc{CMS}& \textsc{CaD}\\\hline  
		 chnl-005x006.shuffled&0.007&0.006&0.004&0.007&0.008&0.004&0.007&0.004\\
chnl-010x011.shuffled&49.896&142.257&0.019&0.009&0.009&0.008&0.008&0.006\\
chnl-015x017.shuffled&600&600&0.075&0.01&0.01&0.017&0.01&0.076\\
chnl-020x021.shuffled&600&600&0.186&0.014&0.019&0.034&0.014&0.014\\
chnl-025x050.shuffled&600&600&2.257&0.027&0.064&0.247&0.027&0.035\\
chnl-030x031.shuffled&600&600&0.896&0.022&0.045&0.116&0.023&0.028\\
chnl-040x041.shuffled&600&600&3.01&0.035&0.159&0.312&0.035&0.05\\
chnl-045x050.shuffled&600&600&6.135&0.05&0.085&0.575&0.048&0.202\\
chnl-050x060.shuffled&600&600&13.605&0.065&0.205&1.102&0.073&0.122\\
chnl-050x070.shuffled&600&600&21.077&0.082&0.166&1.621&0.084&0.232\\
		\hline
	\end{tabular}
	\caption{Individual results for channel routing instances.} \label{fig:ind:channel}
	\end{figure}
\begin{figure}
\centering
\small
	\begin{tabular}{|l|c|c|ccc|ccc|}\hline
		 & \textsc{CMS} & \textsc{CaD} & \multicolumn{3}{c|}{\textsc{BreakID}} & \multicolumn{3}{c|}{\textsc{satsuma}}\\
		 instance & solve & solve & prep & \textsc{CMS}& \textsc{CaD} & prep & \textsc{CMS}& \textsc{CaD}\\\hline  
		 clqcolor15\_5\_4&240.655&55.712&0.021&1.373&0.504&0.013&0.008&0.007\\
		 clqcolor15\_6\_5&600&375.34&0.03&6.508&2.2&0.018&0.009&0.006\\
		 clqcolor16\_5\_4&454.033&63.361&0.025&1.821&0.527&0.015&0.008&0.005\\
		 clqcolor16\_6\_5&600&479.848&0.034&6.745&2.531&0.019&0.009&0.006\\
		 clqcolor17\_5\_4&600&74.222&0.027&2.495&0.945&0.016&0.008&0.005\\
		 clqcolor17\_6\_5&600&600&0.038&17.161&5.054&0.021&0.009&0.006\\
		 clqcolor18\_5\_4&600&86.875&0.031&2.416&0.939&0.017&0.008&0.005\\
		 clqcolor18\_6\_5&600&600&0.044&26.457&7.465&0.024&0.009&0.006\\
		 clqcolor19\_5\_4&600&318.293&0.036&4.331&1.655&0.019&0.008&0.005\\
		 clqcolor20\_5\_4&600&347.361&0.041&7.539&3.143&0.021&0.009&0.006\\
		 clqcolor20\_6\_5&600&600&0.058&50.674&19.377&0.029&0.01&0.007\\
		 clqcolor23\_6\_5&600&600&0.082&228.66&83.039&0.039&0.01&0.008\\
		 clqcolor25\_5\_4&600&446.439&0.073&23.78&12.594&0.033&0.01&0.007\\
		 clqcolor25\_7\_6&600&600&0.139&600&600&0.064&0.017&0.015\\
		 clqcolor25\_8\_7&600&600&0.215&600&600&0.086&0.243&0.062\\
		 clqcolor25\_9\_8&600&600&0.25&600&600&0.109&6.12&0.36\\
		 clqcolor30\_10\_9&600&600&0.507&600&600&0.212&6.398&2.93\\
		 clqcolor30\_7\_6&600&600&0.235&600&600&0.098&0.02&0.019\\
		 clqcolor30\_8\_7&600&600&0.302&600&600&0.132&0.283&0.062\\
		 clqcolor30\_9\_8&600&600&0.4&600&600&0.168&3.713&0.466\\		 
		\hline
	\end{tabular}
\caption{Individual results for clique coloring instances.} \label{fig:ind:clqcolor}
\end{figure}

\begin{figure}
	\centering
	\scriptsize
	\begin{tabular}{|l|c|c|ccc|ccc|}\hline
		 & \textsc{CMS} & \textsc{CaD} & \multicolumn{3}{c|}{\textsc{BreakID}} & \multicolumn{3}{c|}{\textsc{satsuma}}\\
		 instance & solve & solve & prep & \textsc{CMS}& \textsc{CaD} & prep & \textsc{CMS}& \textsc{CaD}\\\hline  
		 anna.col.11&98.456&81.395&0.063&8.425&2.104&0.023&25.912&30.633\\
		 david.col.11&43.52&36.469&0.034&11.2&13.025&0.014&0.918&0.282\\
		 fpsol2.i.1.col.65&600&600&600&600&600&6.434&600&600\\
		 fpsol2.i.2.col.30&600&600&99.052&600&600&2.157&600&600\\
		 fpsol2.i.3.col.30&600&600&72.893&600&600&2.157&600&600\\
		 games120.col.9&2.768&0.376&0.026&6.116&0.401&0.017&2.179&0.311\\
		 homer.col.13&0.01&0.008&5.018&0.015&0.017&0.22&0.013&0.011\\
		 huck.col.11&19.489&49.519&0.037&0.008&0.008&0.013&0.719&1.119\\
		 inithx.i.1.col.54&600&600&600&600&600&15.321&600&600\\
		 inithx.i.2.col.31&600&600&404.173&600&600&6.097&600&600\\
		 inithx.i.3.col.31&600&600&346.767&600&600&6.004&600&600\\
		 jean.col.10&7.661&2.343&0.028&6.84&1.071&0.013&6.918&1.044\\
		 le450\_15a.col.15&600&600&0.469&600&600&0.475&600&600\\
		 le450\_15b.col.15&600&600&0.459&600&600&0.461&600&600\\
		 le450\_15c.col.15&600&600&0.947&600&600&1.363&600&600\\
		 le450\_15d.col.15&600&600&0.953&600&600&1.364&600&600\\
		 le450\_25a.col.25&600&600&1.102&600&600&0.736&600&600\\
		 le450\_25b.col.25&600&600&1.073&600&600&0.752&600&600\\
		 le450\_25c.col.25&600&600&2.334&600&600&3.205&600&600\\
		 le450\_25d.col.25&600&600&2.306&600&600&3.283&600&600\\
		 le450\_5a.col.5&0.009&0.008&0.066&0.011&0.01&0.028&0.012&0.011\\
		 le450\_5b.col.5&0.01&0.008&0.057&0.01&0.01&0.028&0.012&0.01\\
		 le450\_5c.col.5&0.011&0.011&0.097&0.012&0.013&0.044&0.014&0.014\\
		 le450\_5d.col.5&0.012&0.011&0.096&0.012&0.013&0.043&0.014&0.014\\
		 miles1000.col.42&600&600&1.129&600&600&0.412&600&600\\
		 miles1500.col.73&600&600&12.816&600&600&1.419&600&600\\
		 miles250.col.8&0.199&0.064&0.02&0.304&0.059&0.015&0.271&0.075\\
		 miles500.col.20&600&600&0.162&600&600&0.06&600&600\\
		 miles750.col.31&600&600&0.418&600&600&0.227&600&600\\
		 mulsol.i.1.col.49&600&600&30.688&600&600&0.625&600&600\\
		 mulsol.i.2.col.31&600&600&4.85&600&600&0.42&600&600\\
		 mulsol.i.3.col.31&600&600&4.851&600&600&0.415&600&600\\
		 mulsol.i.4.col.31&600&600&3.655&600&600&0.427&600&600\\
		 mulsol.i.5.col.31&600&600&5.215&600&600&0.422&600&600\\
		 myciel3.col.4&0.006&0.003&0.003&0.006&0.003&0.005&0.008&0.005\\
		 myciel4.col.5&0.039&0.015&0.004&0.009&0.007&0.005&0.009&0.005\\
		 myciel5.col.6&39.334&7.687&0.008&1.814&0.248&0.008&2.673&0.262\\
		 myciel6.col.7&600&600&0.019&600&600&0.025&600&268.653\\
		 myciel7.col.8&600&600&0.063&600&600&0.114&600&600\\
		 queen10\_10.col.10&7.188&3.588&0.052&0.01&0.054&0.136&6.338&1.796\\
		 queen11\_11.col.11&16.983&46.396&0.127&0.012&0.021&0.117&0.016&0.024\\
		 queen12\_12.col.12&53.285&600&0.112&0.015&0.027&0.168&0.023&0.035\\
		 queen13\_13.col.13&600&600&0.291&0.017&0.029&0.251&0.039&0.086\\
		 queen14\_14.col.14&600&600&0.236&0.02&0.029&0.342&0.067&0.056\\
		 queen15\_15.col.15&600&600&0.612&0.024&0.036&0.516&0.067&0.133\\
		 queen16\_16.col.16&600&600&0.473&0.047&0.142&0.692&0.119&0.157\\
		 queen5\_5.col.5&0.006&0.003&0.005&0.007&0.003&0.005&0.011&0.003\\
		 queen6\_6.col.7&3.649&0.407&0.009&0.007&0.004&0.009&0.008&0.004\\
		 queen7\_7.col.7&0.015&0.009&0.014&0.008&0.005&0.013&0.007&0.005\\
		 queen8\_12.col.12&60.928&434.188&0.059&0.012&0.021&0.075&0.02&0.036\\
		 queen8\_8.col.9&600&600&0.024&17.623&5.46&0.029&44.979&12.945\\
		 queen9\_9.col.10&600&600&0.055&600&600&0.103&600&600\\
		 zeroin.i.1.col.49&600&600&91.624&600&600&0.719&600&600\\
		 zeroin.i.2.col.30&600&600&10.754&600&600&0.421&28.394&80.356\\
		 zeroin.i.3.col.30&600&600&10.23&600&600&0.442&600&600\\				   
		\hline
	\end{tabular}
	\caption{Individual results for graph coloring instances.} \label{fig:ind:coloring}
	\end{figure}

\begin{figure}
\centering
\small
\begin{tabular}{|l|c|c|ccc|ccc|}\hline
	 & \textsc{CMS} & \textsc{CaD} & \multicolumn{3}{c|}{\textsc{BreakID}} & \multicolumn{3}{c|}{\textsc{satsuma}}\\
	 instance & solve & solve & prep & \textsc{CMS}& \textsc{CaD} & prep & \textsc{CMS}& \textsc{CaD}\\\hline  
	 fpga10\_11\_uns\_rcr&46.254&64.618&0.019&0.009&0.007&0.008&0.008&0.016\\
	 fpga10\_12\_uns\_rcr&54.591&112.249&0.022&0.009&0.008&0.008&0.008&0.055\\
	 fpga10\_13\_uns\_rcr&131.252&264.064&0.025&0.01&0.007&0.009&0.008&0.006\\
	 fpga10\_15\_uns\_rcr&123.957&335.03&0.033&0.01&0.008&0.01&0.008&0.016\\
	 fpga10\_20\_uns\_rcr&143.731&136.936&0.061&0.01&0.008&0.014&0.009&0.062\\
	 fpga11\_12\_uns\_rcr&316.177&600&0.025&0.009&0.007&0.009&0.008&0.006\\
	 fpga11\_13\_uns\_rcr&600&600&0.029&0.01&0.008&0.009&0.008&0.006\\
	 fpga11\_14\_uns\_rcr&600&600&0.033&0.01&0.007&0.01&0.008&0.061\\
	 fpga11\_15\_uns\_rcr&600&600&0.037&0.011&0.01&0.011&0.009&0.018\\
	 fpga11\_20\_uns\_rcr&600&600&0.07&0.011&0.009&0.016&0.009&0.007\\			   
	\hline
\end{tabular}
\caption{Individual results for \textsf{fpga} instances.} \label{fig:ind:fpga}
\end{figure}

\begin{figure}
\centering
\small
\begin{tabular}{|l|c|c|ccc|ccc|}\hline
	 & \textsc{CMS} & \textsc{CaD} & \multicolumn{3}{c|}{\textsc{BreakID}} & \multicolumn{3}{c|}{\textsc{satsuma}}\\
	 instance & solve & solve & prep & \textsc{CMS}& \textsc{CaD} & prep & \textsc{CMS}& \textsc{CaD}\\\hline  
	 gus-md5-04&1.486&1.806&0.624&1.54&2.966&0.507&1.199&1.124\\
	 gus-md5-05&5.103&7.185&0.631&5.119&9.398&0.522&2.815&4.041\\
	 gus-md5-06&15.274&30.497&0.628&18.953&27.612&0.57&8.914&14.64\\
	 gus-md5-07&49.704&37.75&0.628&56.388&34.524&0.538&29.281&27.701\\
	 gus-md5-09&273.206&149.676&0.639&271.201&183.121&0.561&228.512&166.89\\
	 gus-md5-10&600&506.239&0.636&600&518.947&0.543&570.16&357.478\\
	 gus-md5-11&600&600&0.638&600&600&0.576&600&600\\
	 gus-md5-12&600&600&0.643&600&600&0.568&600&600\\
	 gus-md5-14&600&600&0.639&600&600&0.568&600&600\\
	 gus-md5-15&600&600&0.638&600&600&0.549&600&600\\
	 gus-md5-16&600&600&0.645&600&600&0.531&600&600\\					
	\hline
\end{tabular}
\caption{Individual results for \textsf{md5} instances.} \label{fig:ind:md5}
\end{figure}

\begin{figure}
\centering
\small
\begin{tabular}{|l|c|c|ccc|ccc|}\hline
	 & \textsc{CMS} & \textsc{CaD} & \multicolumn{3}{c|}{\textsc{BreakID}} & \multicolumn{3}{c|}{\textsc{satsuma}}\\
	 instance & solve & solve & prep & \textsc{CMS}& \textsc{CaD} & prep & \textsc{CMS}& \textsc{CaD}\\\hline  
	 hole005&0.008&0.004&0.003&0.007&0.003&0.004&0.007&0.004\\
	 hole007&0.273&0.037&0.005&0.007&0.003&0.004&0.007&0.003\\
	 hole010&32.383&29.724&0.007&0.007&0.004&0.004&0.007&0.004\\
	 hole012&600&600&0.01&0.008&0.005&0.005&0.008&0.004\\
	 hole015&600&600&0.019&0.008&0.005&0.005&0.008&0.005\\
	 hole020&600&600&0.048&0.01&0.008&0.008&0.008&0.007\\
	 hole030&600&600&0.217&0.015&0.017&0.02&0.013&0.024\\
	 hole050&600&600&1.657&0.073&0.102&0.083&0.032&0.113\\
	 hole075&600&600&11.027&0.12&0.256&0.315&0.098&0.307\\
	 hole100&600&600&50.373&0.17&0.443&0.832&0.168&0.529\\					 
	\hline
\end{tabular}
\caption{Individual results for pigeonhole principle instances.} \label{fig:ind:php}
\end{figure}

\begin{figure}
\centering
\small
\begin{tabular}{|l|c|c|ccc|ccc|}\hline
	 & \textsc{CMS} & \textsc{CaD} & \multicolumn{3}{c|}{\textsc{BreakID}} & \multicolumn{3}{c|}{\textsc{satsuma}}\\
	 instance & solve & solve & prep & \textsc{CMS}& \textsc{CaD} & prep & \textsc{CMS}& \textsc{CaD}\\\hline  
	 ram3\_3\_6&0.008&0.003&0.003&0.007&0.003&0.004&0.007&0.005\\
	 ram3\_4\_9&0.283&0.075&0.004&0.008&0.004&0.004&0.007&0.005\\
	 ram3\_5\_14&600&600&0.023&1.59&0.346&0.01&0.01&0.01\\
	 ram3\_6\_18&600&600&0.302&600&600&0.09&0.169&0.145\\
	 ram3\_7\_23&600&600&10.253&600&600&2.349&446.697&145.927\\
	 ram4\_4\_18&600&600&0.054&600&600&0.075&600&600\\
	 ram4\_5\_25&600&600&1.131&600&600&0.228&600&600\\					  
	\hline
\end{tabular}
\caption{Individual results for Ramsey instances.} \label{fig:ind:ramsey}
\end{figure}

\begin{figure}
	\centering
	\small
	\begin{tabular}{|l|c|c|ccc|ccc|}\hline
		 & \textsc{CMS} & \textsc{CaD} & \multicolumn{3}{c|}{\textsc{BreakID}} & \multicolumn{3}{c|}{\textsc{satsuma}}\\
		 instance & solve & solve & prep & \textsc{CMS}& \textsc{CaD} & prep & \textsc{CMS}& \textsc{CaD}\\\hline  
		 Urq3\_5&1.145&6.083&0.008&0.007&0.004&0.007&0.007&0.004\\
		 Urq4\_5&0.883&291.65&0.014&0.007&0.004&0.008&0.147&0.087\\
		 Urq5\_5&0.715&600&0.039&0.008&0.005&0.017&0.085&0.108\\
		 Urq6\_5&0.624&600&0.098&0.008&0.005&0.027&0.016&0.01\\
		 Urq7\_5&0.626&600&0.196&0.008&0.005&0.041&0.018&0.01\\
		 Urq8\_5&0.613&600&0.487&0.009&0.006&0.089&0.121&0.094\\							   
		\hline
	\end{tabular}
	\caption{Individual results for \textsf{urquhart} instances.} \label{fig:ind:urquhart}
	\end{figure}

\end{document}